\documentclass[journal,draftcls,onecolumn,11pt,twoside]{IEEEtran}
\usepackage{amsmath,amssymb,amsthm,amsfonts}
\usepackage{textcomp}
\usepackage{graphicx}
\usepackage{xcolor}
\usepackage{cite} 
\usepackage{epstopdf}
\usepackage{mathtools}
\usepackage{bibunits}
\usepackage{enumitem}
\usepackage{subfig}
\usepackage{caption}
\usepackage[font={color=blue},figurename=Fig.]{caption}
 
\usepackage{flushend}

\let\emph\textit

\def\BibTeX{{\rm B\kern-.05em{\sc i\kern-.025em b}\kern-.08em
	T\kern-.1667em\lower.7ex\hbox{E}\kern-.125emX}}

\newtheorem{thm}{Theorem}

\newtheorem{cor}{Corollary}
\newtheorem{lemma}{Lemma}
\newtheorem{remark}{Remark}

\newcommand{\vect}[1]{\boldsymbol{#1}}
\DeclarePairedDelimiter\norm{\lVert}{\rVert}

\begin{document}

\title{ Downlink Secrecy Rate of One-Bit  Massive MIMO System with Active Eavesdropping}
\author{M.~A.~Teeti
	%		Qiang~Ni,~\IEEEmembership{Senior Member,~IEEE,} and Yingzhuang~Liu%
	\thanks{M. Teeti is with the Department of Communication Engineering, East China University of Technology, Nanchang, 330013, China (e-mail: teeti.moh@gmail.com)}
}%
\maketitle

\begin{abstract}
In this study, we consider the physical layer security in the downlink of a Massive MIMO system employing one-bit quantization at the base station (BS).  We assume an active eavesdropper that attempts to spoiling the channel estimation acquisition at the BS for a legitimate user, whereas overhearing on downlink transmission. We consider the two most widespread methods for degrading the eavesdropper's channel, the nullspace artificial noise (NS-AN) and random artificial noise (R-AN). Then, we present a lower bound on the secrecy rate and asymptotic performance, considering zero-forcing beamforming (ZF-BF) and maximum-ratio transmission beamforming (MRT-BF). Our results reveal that even when the eavesdropper is close enough to the intercepted user, a positive secrecy rate --which tends to saturation with increasing the number of BS antennas $N$---is possible, as long as the transmit power of eavesdropper is less than that of the legitimate user during channel training. We show that ZF-BF with NS-AN provides the best performance. It is found that MRT-BF and ZF-BF are equivalent in the asymptotic limit of $N$ and hence the artificial noise technique is the performance indicator. Moreover, we study the impact of \emph{power-scaling law} on the secrecy rate. When the transmit power of BS is reduced proportional to $1/N$, the performance is independent of artificial noise asymptotically and hence the beamforming technique is the performance indicator. In addition, when the BS's power is reduced proportional to $1/\sqrt{N}$, all combinations of beamforming and artificial noise schemes are equally likely asymptotically, independent of quantization noise. We present various numerical results to corroborate our analysis.
\end{abstract}


\IEEEpeerreviewmaketitle

\let\incfile\include
\renewcommand{\include}[1]{\directlua{tex.ConvertToSpace("#1")}}

%==========================================================
\section{Introduction}
\label{sec:introduction}
%==========================================================

Information  secrecy in Massive multiple-input multiple-output (MIMO) system—as a key technology for fifth-generation networks—has been a critical issue that spurred widespread interest \cite{Zhu2014,7108032,7491252,Kapetanovic2015,Goel2008,7574326}. One challenge in Massive MIMO lies in the increase in hardware complexity and energy consumption \cite{7000981} due to a large number of antennas at the base station (BS). In recent years, there has been a growing interest in replacing the high-resolution analog-to-digital converters (ADCs) and digital-to-analog converters (DACs) with low-resolution ADCs and DACs. The extreme case of one-bit ADC/DAC  has been gaining much attention \cite{7586074, 6804238, 7094595, T2018} because of the considerable design simplicity offered to the physical layer and negligible energy consumption. With this in mind, it is of interest to understand the secrecy capability of Massive MIMO employing one-bit quantization, which is the aim of this work.  

In a major advance in 1949, Claude Shannon \cite{Shannon1948} established the information-theoretic basis of communication secrecy of cryptographic systems. In classical security, the transmitter often shields the private message by a means of shared-key cryptographic techniques carried out at the logical layers of the network. Typically, the encryption key is very long and computationally demanding. In addition, it is susceptible to interception by powerful adversaries, especially in a wireless environment. Consequently, a key exchange between two legitimate parties becomes infeasible in dynamic wireless networks with nodes of limited resources. To tackle this problem, physical layer security provides an alternative or a complement to classical cryptography, which exploits the statistical differences between the channel of the legitimate receiver and that of the eavesdropper to guarantee secrecy. 

The first information-theoretic approach to physical layer security dates back to Wyner's work \cite{wyner78} on the degraded Gaussian wiretap channel. Later, Csiszar and Korner \cite{Korner1978} generalized Wyner's work to the non-degraded wiretap channel. In the preceding works of Wyner, Csiszar, and Korner,  it was shown that when the channel of the legitimate receiver is more capable (less noisy) than that of the eavesdropper, secure communication is possible with no need for classical cryptography. The maximal rate at which the transmitter and legitimate receiver can communicate securely is limited by the \emph{secrecy capacity}, defined as the maximal of the difference between the channel mutual information of the legitimate receiver and that of the eavesdropper.

In the literature, \emph{passive attack} refers to the situation where an eavesdropper is concealing himself and thus only eavesdropping on the confidential transmission. On the other hand, \emph{active attack} refers to the situation where an eavesdropper is not only eavesdropping on the confidential transmission but also jamming the transmission. In the literature, many attempts have been made \cite{Zhu2014,7167717, Zhu2016, 7574326,7926385} to study the impact of the passive attack in Massive MIMO systems under different scenarios. One common thing among most of the above works and others in the literature is the use of \emph{artificial noise} to degrade the eavesdropper channel \cite{Goel2008} and hence improve security. Most of the above works focus on the careful design of data beamforming (or precoding) and artificial noise. In the literature, two artificial noise techniques are widely used, the nullspace artificial noise (NS-AN) and random artificial noise (R-AN) \cite{Goel2008}. With NS-AN, the artificial noise is made aligned with the nullspace of the channel of the legitimate user while with the R-AN, the artificial noise is generated randomly.

Massive MIMO has been considered as one of the key technologies enabling green communication for its ability to scale down the transmitted power while maintaining a minimum quality of service to each user in the system. Thus minimizing power (transmitted and circuit power consumption \cite{1522052}) while achieving secure communication simultaneously is of great importance \cite{7812623}. Besides, one of the constrains in wireless communication is the limited battery life of wireless devices \cite{7462485,7962221, zhu2019,zhu2018}.  Thus maximizing the energy harvested while satisfying the requirement of secure communication turns out to be of importance in current and future networks.  Z.  Zhu et al., \cite{zhu2018} studied the information secrecy under the simultaneous wireless information and power transfer (SWIPT) MIMO system, where the authors proposed a low-complexity iterative algorithm to optimize the design of beamforming to maximize the harvested energy and meeting secrecy rate requirements simultaneously. In \cite{7962221} a joint optimization of beamforming and AN in a multiuser MIMO system is considered under target secrecy rate and transmit and harvested power constraints.

Stemming from the fact that meeting physical layer security in the information-theoretic sense may give rise to a significant loss in data rate, Bin Chen et al., \cite{7491252} considers a cryptographic-like scheme to achieve security in Massive MIMO system in the presence of a powerful eavesdropper. In \cite{7491252}, the message symbols are randomly phase rotated while this phase rotation is only available at the legitimate receiver through downlink training with a small amount of overhead. There, it is shown that when the BS is equipped with a sufficiently large number of antennas, we guarantee secure communication with high probability.

It is well-known that the promising gains of Massive MIMO systems \cite{Lu2014, 6375940, 6457363, 7091922}  are affected by \emph{pilot contamination} \cite{5595728,7833554}, whether resulting from pilot reuse \cite{5595728} in multi-cellular networks or \emph{pilot attack} \cite{6151778, Wu2016, Kapetanovic2015} created intentionally by an active eavesdropper. In fact, the pilot attack can cause serious degradation of the secrecy rate since the beamformed signal in the downlink will be partly aligned with the direction of eavesdropper's channel, thereby increasing the information leakage. This situation becomes more pronounced when the pilot attack is severe, under which no positive secrecy rate is possible. 

Many attempts \cite{Yuksel2015, 7010928, 7750607, Do2018, 8335124} with the purpose of detecting and combating pilot attack in Massive MIMO have been done. Yuksel et al., \cite{Yuksel2015} showed that pilot attack can be eliminated asymptotically as the size of the pilot set (which is assumed known to everyone) is increased as long as users select their pilots randomly. Q. {Xiong} et al., \cite{7010928} proposed an efficient energy-based detector to identify a pilot attack without the knowledge of the channel state information (CSI). T. T. {Do} et al., \cite{7750607} considered a single-user uplink Massive MIMO and studied two anti-jamming strategies based on pilot re-transmission and pilot adaptation technique. R. F. {Schaefer} et al., \cite{8335124} studied the secrecy and pilot attack detection in a single-cell Massive MIMO in the presence of a single-antenna eavesdropper. There, it is shown that the secrecy rate can drop to zero as the power of eavesdropper is increased. Tan et al., \cite{Do2018} considered pilot jamming in the uplink and proposed jamming-resistant approach using unused pilot and pilot hopping to estimate the jamming channel. With zero-forcing type receiver, it is shown in \cite{Do2018} that we can greatly enhance the robustness of the Massive MIMO uplink against jamming attacks. 

In multicell multiuser Massive MIMO systems, although pilot contamination resulting from reuse of pilot across the network can be alleviated through coordination between different BSs with low overhead \cite{Marzetta2014} (also see \cite{8246979} for interference alignment-based approach), however, the pilot attack can still pose a real performance risk. Wu et al., \cite{Wu2016} considered an active eavesdropper armed with multiple antennas, and presented signal design using beamforming based on maximum-ratio transmission and NS-AN technique under the correlated channel. They showed that the NS-AN can benefit from the highly correlated channels, enabling secure communication; however, this is not the case when the channel is weakly correlated or independent and identically distributed (i.i.d.). To overcome the limitation in \cite{Wu2016}, the authors in \cite{Wu2018} considered pilot-data exploitation for CSI acquisition. They showed that decreasing the legitimate user's power causes its received signal to lie in a different eigenspace as that of the eavesdropper in the asymptotic limit of data length, thus mitigating the effect of a strong pilot attack. 

Using low-resolution ADCs/DACs at the BS in Massive MIMO can substantially simplify the physical layer and reduce energy consumption, particularly when the one-bit quantization is considered. A related challenge is the design of the channel estimator and the precoder \cite{Jacobsson2017,7946265} which turns to be not trivial as the quantization can break the structure of the beamforming matrix. This challenge can exacerbate when a pilot attack is present in the system.  In \cite{Zhao2016} the design of artificial noise is investigated in a simple scenario of a multiple-antenna system under the constraint of a few RF chains at the BS, considering a passive eavesdropper and perfect CSI at the BS. The impact of hardware impairment  (such as phase noise and amplified receiver noise) on secrecy in Massive MIMO is studied in \cite{Zhu2017} and hence both the uplink training and the design of artificial noise are optimized to enhance secrecy under a passive eavesdropper. More recently, a low-resolution Massive MIMO system with multiple-antenna passive eavesdropper was studied in \cite{8626548}. With perfect CSI  assumed available at the BS, it was shown that quantization noise gives rise to the increase in secrecy rate.  

The main limitation of the previous studies on the secrecy of Massive MIMO system with quantization or limited RF chains at the BS is the focus on passive attack scenarios with the assumption of perfect CSI at the BS. As far as quantization is concerned, the assumption of perfect CSI becomes inaccurate even in the absence of pilot contamination and in particular,  the perfect CSI is unjustified when the extreme one-bit quantization case is considered.  Also, of even greater importance is the impact of active eavesdropping on secrecy in quantized Massive MIMO systems, which is not well understood in the literature. In this work, we will particularly study the one-bit quantized Massive MIMO system with an active eavesdropper, and investigate its secrecy performance under various beamforming and artificial noise techniques.  

As a first step toward understanding the potential secrecy in such quantized systems, we will investigate only the performance under the zero-forcing and maximum ratio combining (or matched filtering) beamformers combined with two widely used techniques for degrading the quality of the eavesdropper's channel: nullspace and random artificial noise techniques.

%==========================================================
\subsection{Contributions}
%==========================================================

We summarize the main contributions of this work as
follows:
\begin{enumerate}[wide, labelwidth=!, labelindent=0pt]
	\normalfont{
	\item  We derive lower bounds on secrecy rate under various beamforming and artificial noise schemes, and an asymptotic performance analysis (when the number of BS antennas $N \to \infty$)  is given.
	\item We show analytically (as $N\to \infty$) a threshold on the transmit power ratio between the eavesdropper and intercepted user below which a positive secrecy rate is possible. As a result, when the eavesdropper is near enough to the intercepted user, secure communication turns to be difficult (if not impossible) when the transmit ratio is close to 1. This result is confirmed by simulation of a practical scenario. 
	\item We show that when there is no power scaling at the BS (i.e., power is not scaled down by the number of BS antennas),  NS-AN technique outperforms R-AN technique, regardless of beamforming technique, as $N \to \infty$.
	\item We show that when the power at the BS is reduced proportional to $1/N$, the zero-forcing beamforming (ZF-BF) outperforms maximum-ratio transmission beamforming (MRT-BF), regardless of artificial noise. Further, when the power is reduced proportional to $1/\sqrt{N}$ all schemes (any combination of beamforming and artificial noise techniques) are asymptotically equivalent and also quantization noise is irrelevant.}
\end{enumerate}

\subsection{Outline}
We organize the rest of the paper as follows. Section \ref{sec2} introduces signal models in uplink and downlink and we discuss channel estimation. Section \ref{sec3} presents the design of downlink beamforming and artificial noise. Also, we show the analysis of information rates, the main results and specializing the main results to an unquantized system and a quantized system with passive eavesdropping. In Section \ref{sec4}, we present the asymptotic performance comparison and we derive the condition under which secure communication is possible. In Section \ref{sec5}, we present some numerical examples to verify our analytical results and Section \ref{sec6} concludes this work. 

\subsection{Notation}
Throughout this paper we use the superscript $T$ to denote transposition and the superscript $H$ to denote hermitian transpose, $E [\cdot]$ and $\operatorname{Var} (\cdot)$ denote the expected value and variance of a random variable, respectively. Boldface capital letter $\vect{X}$ denotes a random matrix, boldface small letter $\vect{x}$ denotes a random vector, small letter in normal font $x$ denotes a scalar random variable, big letter in normal font $X$ is typically used to denote a system parameter. We denote by $\norm{\vect{x}}$ the Euclidean norm of a vector $\vect{x}$, $[\vect{X}]_i$ denotes the $i$-th diagonal entry of a matrix $\vect{X}$, $\operatorname{diag}(a_1,a_2,\cdots)$ denotes a diagonal matrix with $a_1,a_2,\cdots$ comprise its diagonal, and $\operatorname{tr}(\vect{\cdot})$ denotes the matrix trace. The pointwise operations $ \log(\cdot), \operatorname{sign}(\cdot), \Re\{\cdot\}$ and $\Im\{\cdot\}$ denote the logarithm to base 2, sign function, real and imaginary parts of a scalar, vector or matrix, respectively.

%==========================================================
\section{Signal Model and Channel Estimation}
 \label{sec2}
 %==========================================================
 We consider the downlink of a single-cell Massive MIMO system with one-bit ADCs/DACs employed at the BS. We assume that the BS has $N$ antennas, serving $K$ single-antenna users ($K\ll N$) in the same time-frequency resource block. We assume the communication system operates in the time-division duplex (TDD) mode. Further, we assume a single-antenna active eavesdropper who attacks the communication between a legitimate user and the BS by contaminating its CSI acquisition at the BS during channel training and overhearing on the downlink transmission.

We consider Rayleigh block-fading for both BS-users and BS-eavesdropper channels with coherence time $T_c$. Within each block, the channel remains constant over $T_c$ symbol intervals and changes independently from one block to another. The composite small-scale fading channel between all legitimate users and the BS is denoted by $\vect{H} = [\vect{h}_1,\vect{h}_2,\cdots,\vect{h}_K] \in \mathcal{C}^{N\times K}$ and $\vect{g} \in \mathcal{C}^{N}$ represents the small-scale fading channel between the eavesdropper and the BS. The $(n,j)$-th component of $\vect{H}$, denoted $h_{n j}$, represents the propagation gain between the $n$-th BS antenna and user $j$, whereas $g_n$ denotes the propagation gain between the BS antenna $n$ and the eavesdropper. Both $\vect{H}$ and $\vect{g}$ comprise i.i.d. complex Gaussian random variables, each with zero-mean and unit variance.  Further, we denote by  $\beta_j$ \& $\beta_{e}$ the large-scale fading coefficients associated with legitimate user $j$ and the eavesdropper, respectively. We assume all large-scale fading coefficients change slowly in order of several $T_c$ intervals and hence assumed available to everyone. Since we are interested in the downlink rate, we divide the coherence time into two parts; one for training (over $\tau$ symbol intervals) and the other for downlink transmission (over $T_c-\tau$ symbol intervals).

%==========================================================
\subsection{Uplink signal model}
\label{sec2.1}
%==========================================================
At the start of communication, all legitimate users in the system send mutually orthogonal pilot sequences of length $\tau$ symbols in the uplink for channel estimation at the BS, whereas the eavesdropper concurrently transmits the same pilot sequence of user $k$ (intercepted user) to impair its channel acquisition at the BS (see Fig. \ref{sys_model1}).  We denote by $\vect{\Psi} = [\vect{\psi}_1,\vect{\psi}_2,\cdots,\vect{\psi}_K]^T \in \mathcal{C}^{K \times \tau}$ the pilot matrix satisfying $\vect{\Psi} \vect{\Psi}^H = \tau I_K$. The $j$-th pilot sequence is expressed as $\vect{\psi}_j = [\psi_j(1),\psi_j(2),\cdots, \psi_j(\tau)]^T \in \mathcal{C}^{\tau}$ where $\psi_j(t)$ is a discrete-time pilot symbol sent from user $j$ at time $t$. For simplicity of analysis, there is no loss of generality in assuming the pilot symbols $\{ \psi_j(t)\}$ to have unit modulus, i.e., $|\psi_j(t)|^2=1$.

\begin{figure}[!ht]
	\centering
	\includegraphics[width=6in]{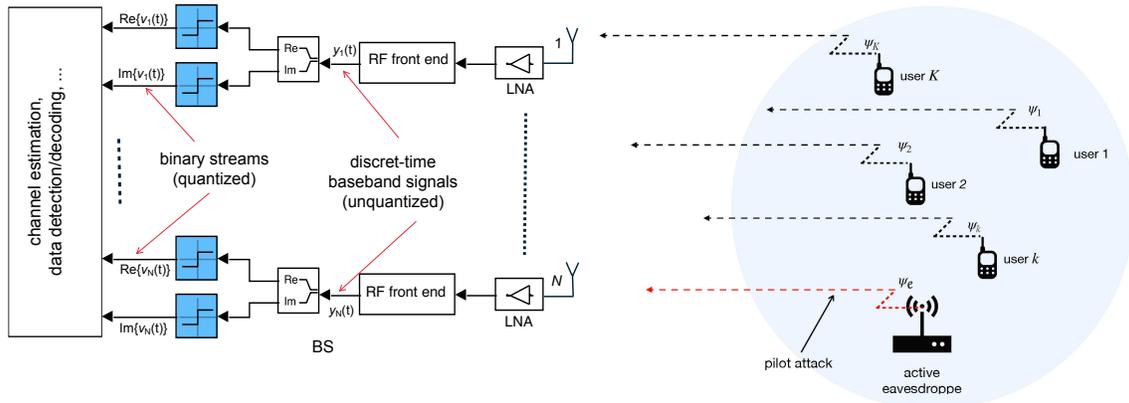}
	\caption{ System model in uplink. }
		\label{sys_model1}
\end{figure}

%\Figure[htp](topskip=0pt, botskip=0pt, midskip=0pt)[scale=0.6]{./figures/sys1.eps}
%{System model in uplink
%	\label{sys_model1}	}

Thus, the discrete-time received signal at the BS during $\tau$ symbol intervals can be written as
\begin{equation} \label{eq0}
\vect{Y} = \sum_{j=1}^K \sqrt{p_j^{\prime}} \vect{h}_j \vect{\psi}_j^T + \sqrt{ p_e^{\prime}} \vect{g} \vect{\psi}_k^T + \vect{Z}  
\end{equation}
where $p_j^{\prime}$ and $p_e^{\prime}$ are the average received power at the BS from user $j$ and eavesdropper, respectively, i.e.,
\begin{subequations}
	\begin{equation}
	p_j^{\prime}= \beta_j p_ j 
	\end{equation}
	\begin{equation}
	p_e^{\prime} = \beta_e p_e
	\end{equation}
\end{subequations}
where $p_j$ and $p_e$ are the average transmit powers of user $j$ and eavesdropper, respectively. The matrix $\vect{Z} \in \mathcal{C}^{N \times \tau}$ denotes a complex additive white Gaussian noise (AWGN) with i.i.d.  $\mathcal{CN}(0,1)$ entries. 

Let $\vect{y}^T_n = [ y_n(1), y_n(2),\cdots, y_n(\tau)]$ be the $n$-th row of $\vect{Y}$  (i.e., signal received by BS antenna $n$) and $\vect{z}^T_n = [ z_n(1), z_n(2), \cdots, z_n(\tau)]$ be the $n$-th row of $\vect{Z}$. Thus $\vect{y}^T_n$ can be expressed as 

\begin{equation} \label{eq:row_n_model}
\vect{y}^T_n  = \sum_{j=1}^K \sqrt{p_j^{\prime}} {h}_{nj}  \vect{\psi}_j^T + \sqrt{ p_e^{\prime}} {g}_n \vect{\psi}_k^T +  \vect{z}^T_n.
\end{equation}
We observe that the row vectors of $\vect{Y}$ are independent of each other and have the common covariance matrix given by
\begin{equation}
\label{Covariance_matrix_first_row}
\vect{C}_{\vect{y}_n} =\vect{\Psi}^H {\vect{P}^{\prime}} \vect{\Psi} + p_e^{\prime} \vect{\psi}_k^{\ast} \vect{\psi}_k^{T} + \vect{I}_{\tau}.
\end{equation}
where ${\vect{P}^{\prime}}=\operatorname{diag}(p_1^{\prime},\cdots,p_K^{\prime})$ is a diagonal matrix. Therefore, it suffices to consider the signal model \eqref{eq:row_n_model} for our analysis. From \eqref{eq:row_n_model}, the $t$-th component of $\vect{y}^T_n $ is given by
\begin{equation} \label{eq1}
y_n(t) =  \sum \limits_{j= 1}^{K}  \sqrt{p_{j}^{\prime}} {h}_{n j} \psi_j (t) +
\sqrt{p_e^{\prime}}  {g}_{n}  \psi_k(t) +  {z}_{n}(t).
\end{equation}

Then, the signal after the one-bit quantizer (one-bit ADC) attached to the $n$-th BS antenna is expressed as
\begin{equation} \label{eq_quantized_sig}
v_n(t) = \operatorname{sign} ({y}_{n}(t)), \hspace{5pt} t=1,2,\cdots, \tau
\end{equation}  
where  $\operatorname{sign} (\cdot)$ is the sign function which yields the sign of the real and imaginary parts of $y_n(t) $ independently. Here we assume a zero-threshold quantizer. Accordingly, the constellation of the quantized signal corresponds to the quadrature phase-shift keying constellation, i.e.,  $\mathcal{A} = \frac{1}{\sqrt{2}}\{1+j,1-j,-1+j,-1-j\}$.

Because of the non-linearity of \eqref{eq_quantized_sig}, the analysis is difficult. However, since ${y}_{n}(t) $ is a Gaussian random variable, it holds from the Bussgang theorem \cite{Bussgang52} that we may express \eqref{eq_quantized_sig} as a sum of a scaled version of ${y}_{n}(t)$ and an uncorrelated term (quantization noise) \cite{Mezghani2012} \cite{T2018}, i.e.,
%========================================================
\begin{align} \label{eq2} 
{v}_{n}(t)   = \gamma {y}_{n} (t) + {q}_{n} (t)
\end{align}
%========================================================
where $\gamma <1$ is a scaling factor and ${q}_{n}(t)$ is the quantization noise uncorrelated to ${y}_{n} (t)$.  From \eqref{eq2}, $\gamma$ is obtained by the linear minimum mean squared error (LMMSE) solution, i.e., $ \gamma = E[y_n^{\ast} (t) v_n(t)] / \sigma_{{y}}^2$, resulting in quantization noise with minimum variance. From \cite{Bussgang52} (see also \cite{Mezghani2012} \cite{T2018}), $E[y_n^{\ast} (t) v_n(t)] =\sqrt{2 	\sigma_{{y}}^2 /\pi} $, where $\sigma_{{y}}^2$ is the variance of $ {y}_{n} (t) $. Hence, 
%========================================================
\begin{equation}\label{eq22} 
\gamma=\sqrt{\frac{2}{\pi \sigma_y^2   }}=\sqrt{\frac{2/\pi}{ \sum_{j=1}^K  p_j^{\prime}+ {p}_{e}^{\prime} +1}}.
\end{equation}
%========================================================
Substituting \eqref{eq22} in \eqref{eq2}, the variance of quantization noise is
\begin{align}\label{eq222}  \nonumber
\sigma_q^2 &= E[|v_n(t)|^2] - \gamma^2 E[|y_n(t)|]\\
&=1-{2}/{\pi} \approx 0.3634.
\end{align}

Stacking the successive symbols  ${v}_{n}(t) (t=1,2,\cdots, \tau)$ in a row vector $\vect{v}_n^T = [v_n(1), v_n(2),\cdots,v_n(\tau) ]$, we obtain
\begin{equation} \label{eq:quantized_signal_vectorform}
\vect{v}_n^T  = \gamma \vect{y}_n^T +\vect{q}_n^T 
\end{equation}
where  $\vect{q}_n^T = [q_n(1), q_n(2),\cdots,q_n(\tau) ]$.

In this work the quantization noise  is assumed uncorrelated \cite{T2018}, i.e., $\vect{C}_{\vect{q}_n} =\sigma_q^2 \vect{I}_{\tau}$. This can be justified as follows.  Using \eqref{eq:quantized_signal_vectorform} the covariance matrix (or correlation) of $\vect{q}_n^T$ can be written as \cite{Mezghani2012}\cite{T2018}
\begin{align} \label{eq:QN_covarinace_matrix}
\nonumber
\vect{C}_{\vect{q}_n} &=\vect{C}_{\vect{v}_n} - \gamma^2 \vect{C}_{\vect{y}_n} \\
&= \frac{2}{\pi}  \sin^{-1} \left[ \vect{\Sigma}_{\vect{y}_n}^{-1/2}  \vect{C}_{\vect{y}_n}   \vect{\Sigma}_{\vect{y}_n}^{-1/2}    \right]- \gamma^2 \vect{C}_{\vect{y}_n} 
\end{align}
where $\vect{C}_{\vect{v}_n} $ is the covariance matrix of ${\vect{v}_n}$ and $\vect{\Sigma}_{\vect{y}_n}$ is the diagonal matrix constructed from the diagonal entries of $\vect{C}_{\vect{y}_n} $. It is seen from \eqref{eq:QN_covarinace_matrix} that when the input of quantizer is uncorrelated (i.e., $\vect{C}_{\vect{y}_n}   \propto  \vect{I}_{\tau}$) so is the quantization noise. Note that the diagonal entries of $\vect{C}_{\vect{y}_n}$ are all equal to $\sigma_y^2 = \sum_{j=1}^K  p_j^{\prime}+ {p}_{e}^{\prime} +1$, while the off-diagonal entries can be expressed as 
\begin{equation}
	\label{eq:ctt}  
\vect{C}_{\vect{y}_n}( t ,t^{\prime}) =\sum_{i=1}^K  p_i^{\prime}\underbrace{\psi_i^{\ast} (t) \psi_i (t^{\prime})}_{e^{j\phi_{i}(t,t^{\prime})}} + p_{e}^{\prime} \psi_k^{\ast} (t) \psi_k (t^{\prime}), t\neq t^{\prime}
\end{equation}

From \eqref{eq:ctt} we observe that when $K$ is sufficiently large, the magnitudes of off-diagonal entries of $\vect{C}_{\vect{y}_n} $ are really much smaller than its diagonal entries, i.e., $\sigma_y^2\gg |\vect{C}_{\vect{y}_n} (t,t^{\prime})|$, due to the sum of a large number of weighted complex exponentials having distinct phases in \eqref{eq:ctt}. Thus we can approximate $\vect{C}_{\vect{y}_n} $ as a diagonal matrix, i.e., $\vect{C}_{\vect{y}_n} \approx \sigma_y^2 \vect{
I}_{\tau}$, leading to $\vect{C}_{\vect{q}_n} \approx \sigma_q^2 \vect{I}_{\tau}$. Finally, without loss of generality, in this work we assume $\tau = K \gg 1$.
   
%   \begin{figure}[!ht]
%   	\centering
%   	\includegraphics[width=6in]{./figures/sys2.eps}
%   	\caption{ System model in downlink. }
% \label{fig:system2}
%   \end{figure}

% \Figure[htp](topskip=0pt, botskip=0pt, midskip=0pt)[scale=0.63]{./figures/sys2.eps}
%{System model in downlink
% \label{fig:system2}}

 %==========================================================
\subsection{Channel estimation}
\label{sec2.2}
%==========================================================

To estimate the propagation gain $h_{nl}$ (from user $l$ to BS antenna $n$), the BS correlates \eqref{eq:quantized_signal_vectorform} with the pilot sequence of user $l$. Hence,
%========================================================
\begin{align}
\label{eq4.1}  \nonumber
\tilde{v}_{l} &:= \frac{1}{\sqrt{\tau}}   \vect{v}_n^T \vect{\psi}_l^{\ast} =\frac{1}{\sqrt{\tau}} \sum \limits_{t= 1}^{\tau} \psi_l^{\ast} (t) v_n(t) \\
&=  \sqrt{  \gamma^2 \tau p_{l}^{\prime}  } h_{n l} + { \sqrt{ \gamma^2 \tau p_e^{\prime}}    {g}_{n} \delta(l-k)} + {\gamma  \tilde{z}_{l}} + { \tilde{q}_{l}}
\end{align}
where  $\tilde{z}_{l}= \vect{z}_n^T \vect{\psi}_l^{\ast}/\sqrt{\tau}$ and $\tilde{q}_{l}=\vect{q}_n^T \vect{\psi}_l^{\ast}/\sqrt{\tau}$ are zero-mean scalar random variables with variances 1 and $\sigma_q^2$, respectively. 

Using \eqref{eq4.1}  the LMMSE estimate of $h_{n l}$ reads
%========================================================
\begin{align}\label{eq4.2} 
\hat{h}_{n l} = \frac{\gamma \sqrt{  p_{l}^{\prime} \tau}   }{ \gamma^2   p_{l}^{\prime} \tau+  \gamma^2  p_{e}^{\prime} \tau \delta(l-k)+  \gamma^2 +   \sigma_q^2 } \tilde{v}_{l} := \lambda_l \tilde{v}_{l}
\end{align}
%=================================
and therefore the variance of ${\hat{h}_{nl}}$ is
\begin{equation}\label{eq4.3}
\sigma_{\hat{h}_{l}}^2= \frac{\gamma^2 p_{l}^{\prime} \tau}{\gamma^2   p_{l}^{\prime} \tau + \gamma^2  p_{e}^{\prime} \tau \delta(l-k)+  \gamma^2 +   \sigma_q^2 }.
\end{equation} 

Stacking all channel estimates in a matrix form, the  composite channel estimate, denoted $\widehat{\vect{H}}$, can be written  as 
\begin{equation}\label{eq5} 
\widehat{\vect{H}}  ={\vect{V} \vect{\Psi}^H \vect{\Lambda}} /{\sqrt{\tau}}
\end{equation}
where $\vect{\Lambda}= \operatorname{diag}(\lambda_1,\lambda_2,\cdots,\lambda_K) \in \mathcal{R}^{K \times K}$ is a diagonal matrix and  $\vect{V} \in \mathcal{C}^{N \times \tau}$ is the quantized signal corresponding to $\vect{Y}$, where the $(n,t)$-th entry of $\vect{V}$ is defined in \eqref{eq_quantized_sig}. Finally, we remark that the channel estimates $\hat{h}_{n l}$ are treated as i.i.d. $\mathcal{CN} (0, \sigma_{\hat{h}_{l}}^2)$, thanks to the law of large numbers. This follows from the fact that $\tilde{v}_{l}$  is typically comprised of a sum of a large number of random variables. 

%==========================================================
\subsection{Downlink signal model}
\label{sec2.3}
%==========================================================

Over one symbol interval, the BS synthesizes the following signal vector (complex baseband precoded signal):
\begin{equation}
\label{eq5.1} 
\tilde {\vect{x}} = \underbrace{\sqrt{\frac{\theta} {\eta} } \vect{W} \vect{s}}_{\textnormal{ information}} + \underbrace{\sqrt{\frac{\bar{\theta}}{{\zeta}}} \vect{n}}_{\textnormal{artificial noise}}
\end{equation}
where $\vect{s} = [s_1,s_2,\cdots,s_K]^T$ comprises $K$ independent complex Gaussian information symbols, i.e., $s_j \sim \mathcal{CN}(0 ,1)$, $\vect{W }= [\vect{w}_1,\vect{w}_2,\cdots,\vect{w}_K]  \in \mathcal{C}^{N \times K}$ is the precoding (or beamforming) matrix with $\vect{w}_j$  being the $j$-th column of $\vect{W}$, and $\vect{n}=[n_1,n_2,\cdots,n_N]^T \in \mathcal{C}^N$ is a zero-mean complex artificial noise vector generated deliberately to weaken the eavesdropper's channel. In \eqref{eq5.1}, $\eta$ and $\zeta$ are long-term normalization constants given by $\eta= E[\operatorname{tr}(\vect{W} \vect{W}^H)] $ and $ \zeta= E[\norm{\vect{n}}^2] $.  Further, $\theta \in (0,1)$ and $\bar{\theta} = 1-\theta$ are the power fractions allocated to the beamformed signal and artificial noise, respectively. Consequently, we have $E[\norm{\tilde{\vect{x}}}^2] = 1$.
   \begin{figure}[!ht]
	\centering
	\includegraphics[width=6in]{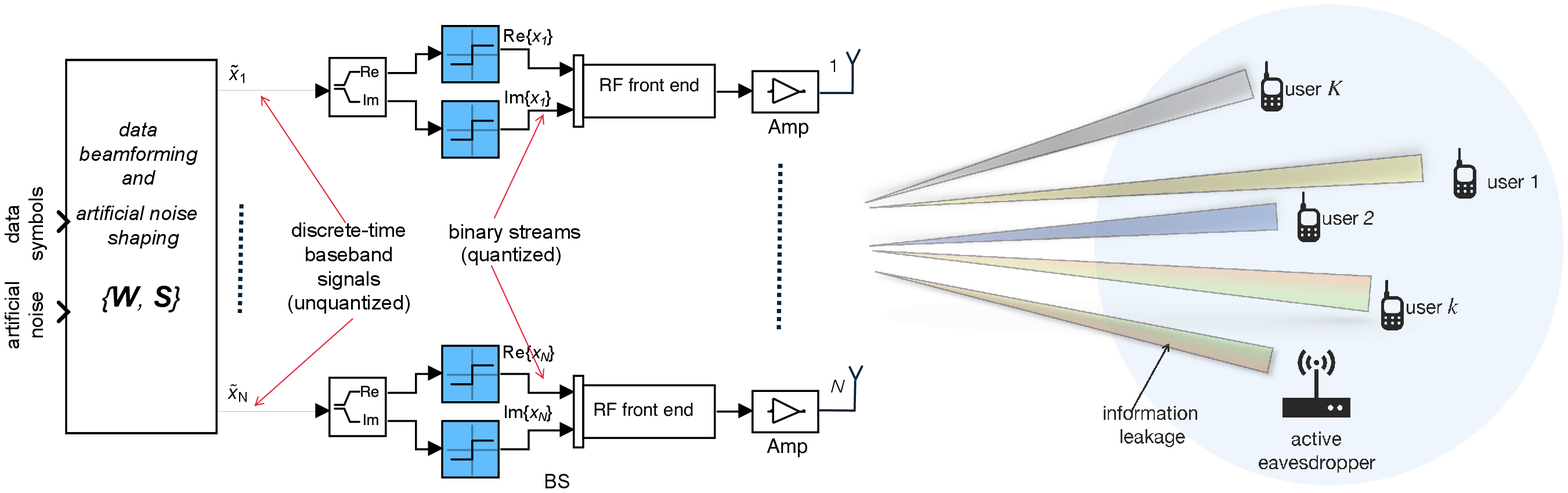}
	\caption{ System model in downlink. }
	\label{fig:system2}
\end{figure}

Then, after the one-bit quantizers at the BS, the signs of the real and imaginary part of $\tilde{\vect{x}}$ are retrieved (pointwise). The complex baseband representation of the transmitted signal is thus given by
\begin{equation}
\label{eq5.2} 
{\vect{x}} = \sqrt{{p_d}/{N}}   \operatorname{sign} (\tilde {\vect{x}} )
\end{equation}
where the scaling factor $\sqrt{p_d/N}$ is introduced to restrict the average transmit power at the BS to $p_d$. The system model in downlink in depicted in Fig. \ref{fig:system2}.

Since $\tilde {\vect{x}}$ is a unit norm vector and we consider the channel matrix $\vect{H}$ drawn from random Gaussian matrix ensembles, the variance of each component of the precoded signal $\tilde {\vect{x}}$ turns to be $\sigma_{\tilde{x}}^2= 1/N$. By linearizing the nonlinear model in \eqref{eq5.2} as we have discussed previously, we can express \eqref{eq5.2} as
%========================================================
\begin{align}
\label{eq6}
 \nonumber
\vect{x} & =\sqrt{\frac{p_d}{N}}   (\bar{\gamma} \tilde {\vect{x}} +\bar{\vect{q}} ) \\ \nonumber
& =\sqrt{\frac{ \theta  \bar{\gamma}^2 p_d   }{N\eta}} \vect{W} \vect{s} + \sqrt{ \frac{\bar{\theta} \bar{\gamma}^2 p_d   }{N \zeta}} \vect{n}+ \sqrt{\frac{p_d}{N}}   \bar{\vect{q}}\\
&=\sqrt{\frac{ 2 \theta  p_d   }{\pi \eta}} \vect{W} \vect{s} + \sqrt{ \frac{2 \bar{\theta}  p_d   }{\pi\zeta}} \vect{n}+ \sqrt{\frac{p_d}{N}}   \bar{\vect{q}}
\end{align}
where $\bar{\vect{q}}$ is the quantization noise, which is assumed to be uncorrelated, i.e., $\vect{C}_{\bar{\vect{q}}}= \sigma_q^2 \vect{I}_N$ and $\bar{\gamma }$ is a scaling factor given by
\begin{equation}
\bar{\gamma}:=\sqrt{\frac{2 }{ \pi \sigma_{\tilde{x}}^2}} = \sqrt{\frac{2N }{\pi}}.
\end{equation}

 For simplicity of notation, we express \eqref{eq6} as
%========================================================
\begin{equation}
\vect{x} =c_1 \vect{W} \vect{s} + c_2 \vect{n}+ c_3 \bar{\vect{q}}
\end{equation}
where $c_1, c_2$ and $c_3$ are, respectively, defined as
\begin{IEEEeqnarray}{CC}
\IEEEyesnumber \label{eq:both} \IEEEyessubnumber*
c_1  & =\sqrt{2 \theta  p_d / \pi \eta} \label{eq:c1} \\
c_2  & =\sqrt{ {2 \bar{\theta} p_d   }/{ \pi \zeta}}  \label{eq:c2} \\
c_3  &= \sqrt{{p_d}/{N}}. \label{eq:c3}
\end{IEEEeqnarray}

\begin{remark}
	\label{remark:rem_UQ}
To retrieve the unquantized signal model from the quantized signal model \eqref{eq6}, we simply replace $p_d$ by $p_d \pi /2$ and let $\sigma_q^2=0$.
\end{remark}
\section{Secrecy capacity analysis}
\label{sec3}
In this section, we establish the achievable rate $\underline{R}_k$ of the intercepted user $k$, and an upper bound on the eavesdropper's rate $\overline{R}_{e}$. We use the underline and overline notation to distinguish between a lower bound and upper bound, respectively. Then the achievable secrecy rate $\underline{R}_s$ is given by \cite{Zhu2014} \cite{Korner1978}
%========================================================
\begin{equation} 
\label{eq:Rs_Def}
\underline{R}_s=
\left [  \underline{R}_k-\overline{R}_{e}  \right ]^{+} 
\end{equation}
where $[A]^{+}=A$ when $A>0$ and $[A]^{+}=0$ when $A<0$.

%==========================================================
\subsection{Data beamforming and artificial noise}
\label{sec3.1}
%==========================================================

In this work, we will consider two classical beamforming techniques: the maximum ratio transmission beamforming (MRT-BF) and zero-forcing beamforming (ZF-BF). Using the channel estimate $\widehat{\vect{H}}$ in \eqref{eq5},  the beamforming matrices of MRT-BF and  ZF-BF are thus given by
\begin{equation}
\label{eq:mrt_zf_bf}
\vect{W} \coloneqq
\begin{cases}
\widehat{\vect{H}}^{\ast} & \text{MRT-BF}, \\
\widehat{\vect{H}}^{\ast} \left(\widehat{\vect{H}}^T \widehat{\vect{H}}^{\ast}\right)^{-1}  &\text{ZF-BF} .
\end{cases}
\end{equation}

For the artificial noise, the vector $\vect{n}$ in \eqref{eq5.1} is defined by 
\begin{equation} \label{eq_AN}
\vect{n} = \vect{S} \tilde{\vect{n}} 
\end{equation}
where $\vect{S}$ is a shaping matrix and $\tilde{\vect{n}}$ is an $N\times 1$ Gaussian vector with i.i.d. $\mathcal{CN}(0,1)$ components. In this work, we will consider the two widely known designs of artificial noise; the R-AN and NS-AN. In the R-AN approach, the noise is purely random, which points to no specific direction. Therefore, we let $ \vect{S} = \vect{I}_N$, thus $ \vect{n} =\tilde{ \vect{n} }$. When NS-AN approach is used, $\vect{n} \in \operatorname{nullspace}(\widehat{\vect{H}}^T )$, i.e., 
\begin{equation} \label{eq:Nullspace_shaping}
\widehat{\vect{H}}^T \vect{S} = \vect{0}_{K \times K}.
\end{equation}

This means that under unquantized systems, when the channel estimate is sufficiently accurate and both user and eavesdropper channels are not highly correlated, a large part of the nullspace artificial noise will be annihilated at the user while it is observed as a random noise at the eavesdropper, degrading its channel quality. To satisfy \eqref{eq:Nullspace_shaping}, we simply choose $\vect{S}$ to be the orthogonal complement matrix of $\widehat{\vect{H}}^T$, given by $\vect{S} = \vect{I}_N - \widehat{\vect{H}}^{\ast} (\widehat{\vect{H}}^T \widehat{\vect{H}}^{\ast})^{-1}\widehat{\vect{H}}^T$.

 We summarize:\begin{equation}
\label{eq_shaping}
\vect{S} \coloneqq
\begin{cases}
\vect{I}_N &\text{R-AN}, \\
\vect{I}_N - \underbrace{\widehat{\vect{H}}^{\ast} (\widehat{\vect{H}}^T \widehat{\vect{H}}^{\ast})^{-1}\widehat{\vect{H}}^T}_{\vect{P}_{\text{proj}}} &\text{NS-AN} .
\end{cases}
\end{equation}

Next the beamforming and AN normalization constants $\eta$ and $\zeta$ in \eqref{eq5.1} are evaluated as follows. Let the channel estimate $\widehat{\vect{H}}$ be decomposed as
\begin{equation} \label{eq:CE_decomposition}
\widehat{\vect{H}} = \widetilde{\vect{H}} \vect{\Sigma}^{1/2}
\end{equation}
 where $\widetilde{\vect{H}}$ is a random Gaussian matrix with i.i.d. $\mathcal{CN}(0,1)$ components, and $\vect{\Sigma}=\operatorname{diag} (\sigma_{\hat{h}_1}^2,\sigma_{\hat{h}_2}^2,\cdots,\sigma_{\hat{h}_K}^2)$ is a diagonal matrix whose diagonal elements comprise the row vector $(\sigma_{\hat{h}_1}^2,\sigma_{\hat{h}_2}^2,\cdots,\sigma_{\hat{h}_K}^2)$. 
 
Thus, when MRT-BF is used we have that  
\begin{align}\label{eq_norm_mrt}\nonumber
\eta_{\text{mrt}} &= E[\operatorname{tr} (  \vect{W}_{\text{mrt}}  \vect{W}_{\text{mrt}}^H)] \\ \nonumber
& = E[\operatorname{tr} (\vect{\Sigma}^{1/2} \widetilde{\vect{H}}^{T} \widetilde{\vect{H}}^{\ast}  \vect{\Sigma}^{1/2})] \\
&=\operatorname{tr} (\vect{\Sigma}^{1/2} E[ \widetilde{\vect{H}}^{T} \widetilde{\vect{H}}^{\ast}]  \vect{\Sigma}^{1/2})=N \operatorname{tr} (\vect{\Sigma})
\end{align}
and when ZF is used, we can write
\begin{align}\label{eq_norm_zf} \nonumber
\eta_{\text{zf}} &= E[\operatorname{tr} (  \vect{W}_{\text{zf}}^H  \vect{W}_{\text{zf}})] \\ \nonumber
&= E[\operatorname{tr} (\vect{\Sigma}^{-1/2} (\widetilde{\vect{H}}^{T} \widetilde{\vect{H}}^{\ast})^{-1}  \vect{\Sigma}^{-1/2})] \\
&= \operatorname{tr} (\vect{\Sigma}^{-1/2} E[ (\widetilde{\vect{H}}^{T} \widetilde{\vect{H}}^{\ast})^{-1}]  \vect{\Sigma}^{-1/2}) = \frac{\operatorname{tr} (\vect{\Sigma}^{-1})}{N-K}
\end{align}
where in \eqref{eq_norm_zf} we have used $E[ (\widetilde{\vect{H}}^{T} \widetilde{\vect{H}}^{\ast})^{-1}]= (N-K)^{-1} \vect{I}_K$, which follows from the property of the inverse of central Wishart matrix $\widetilde{\vect{H}}^{T} \widetilde{\vect{H}}^{\ast}$ \cite{RMT_2004}.

From \eqref{eq_shaping}, the respective AN normalization constants corresponding to R-AN and NS-AN are
\begin{align}
\zeta_{\text{r-an}}&= E[\operatorname{tr} (\vect{I}_N)]=N \\
\nonumber
\zeta_{\text{ns-an}} & = E[\operatorname{tr} (\vect{I}_N- \vect{P}_{\text{proj}})(\vect{I}_ N-\vect{P}_{\text{proj}})^H] \\ \nonumber
&=E[\operatorname{tr} (\vect{I}_N- \vect{P}_{\text{proj}})] \\
&= N-K.
\end{align}

By substituting the above  derived normalization constants in \eqref{eq:both}, we again rewrite \eqref{eq:both} as  
		\begin{align}
\label{defc1}
{c_1} &=
\begin{cases}
\sqrt{{2 \theta  p_d}/ { \pi N \operatorname{tr}(\vect{\Sigma})}} &\text{MRT-BF}\\
\sqrt{{2 \theta  p_d (N-K)}/ { \pi N \operatorname{tr}(\vect{\Sigma}^{-1})} } & \text{ZF-BF} 
\end{cases}\\
\label{defc2}
{c_2} &=
\begin{cases}
\sqrt{ {2 \bar{\theta} p_d   }/{ \pi N}}  &\text{      if R-AN}, \\
\sqrt{ {2 \bar{\theta} p_d   }/{ \pi (N-K)}}  &\text{   if NS-AN} .
\end{cases}\\
\label{defc3}
{c_3} &= \sqrt{p_d/N}.
\end{align}
%==========================================================
\subsection{Data rates  analysis}
\label{sec3.2}
%==========================================================

The received signal at the intercepted user $k$ is
\begin{align}\label{eq7}  \nonumber
r_k &= \sqrt{\beta_k c_1^2}  \vect{h}_k^T  \vect{w}_k s_k + \sum \limits_{j=1, j\neq k}^K  \sqrt{\beta_k c_1^2}   \vect{h}_k^T  \vect{w}_j s_j   \\
&+ \sqrt{\beta_k c_2^2}  \vect{h}_k^T \vect{S} \tilde{\vect{n}} + \sqrt{\beta_k c_3^2}  \vect{h}_k^T \bar{\vect{q}}  +\nu_k
\end{align} 
and the eavesdropper receives
\begin{align}\label{eq8}  \nonumber
r_{e} &= \sqrt{\beta_{e} c_1^2}  \vect{g}^T\vect{w}_k s_k + \sum \limits_{j=1,  j \neq k}^K  \sqrt{\beta_{e} c_1^2 }  \vect{g}^T  \vect{w}_j s_j \\
&+ \sqrt{\beta_{e} c_2^2 }  \vect{g}^T  \vect{S} \tilde{\vect{n}}  + \sqrt{\beta_{e} c_3^2}  \vect{g}^T \bar{\vect{q}} + {\nu}_{e}
\end{align}
where both $\nu_{k}, {\nu}_{e} \sim \mathcal{CN}(0,1)$, denoting the Gaussian noises at the intercepted user and eavesdropper, respectively. 

To obtain a lower bound on secrecy rate, we shall make two main assumptions that have been considered in the literature, serving as a worst-case scenario \cite{Zhu2014,Wu2016}. First, to obtain a lower bound on rate achievable by the legitimate user, we assume the legitimate user has no access to its channel realization and its beamforming vector, and thus the user utilizes only its knowledge of the long-term statistics of the channel for decoding.  Second, to obtain an upper bound on information leakage, we assume the eavesdropper has access to its channel realizations and the beamforming vector of intercepted user. Further, we assume that the eavesdropper can cancel out all inter-user interference, which is conceivable through collusion of other users with the eavesdropper. 

Therefore, after ignoring the second term in \eqref{eq8}, we rewrite \eqref{eq8} as
%=================================================
\begin{align}\label{eq8.1}   
r_{e} =\sqrt{\beta_{e} c_1^2}  \vect{g}^T\vect{w}_k s_k + \sqrt{\beta_{e} c_2^2 }  \vect{g}^T  \vect{S} \tilde{\vect{n}}  + \sqrt{\beta_{e} c_3^2}  \vect{g}^T \bar{\vect{q}} + {\nu}_{e}
\end{align}
and hence an upper bound on the ergodic information rate leaked to the eavesdropper is given by\footnote{In \eqref{eq10.3} we have treated the quantization noise as Gaussian, which is a technical assumption justified by the law of large numbers. Note that $\vect{g}^T \bar{\vect{q}}$ (third term in \eqref{eq8.1}) is a sum of $N$ (large) random numbers and hence can be well approximated as Gaussian random variable.}  
%======================================== 
\begin{equation} \label{eq10.3}
\overline{R}_{e}=  E \left[ \log \left ( 1 +  c_1^2 \beta_{e} \frac{ \norm{\vect{w}_k^H \vect{g}^{\ast}}^2 }{\sigma_{e}^{2}}\right) \right]
\end{equation}
where $\sigma_{e}^{2}$ is the variance of the effective noise seen by the eavesdropper, given by
\begin{equation} \label{eq_cov_eavesdropper}
\sigma_{e}^{2}  =  c_2^2 \beta_{e} \vect{g}^T \vect{S} \vect{g}^{\ast}+ c_3^2 \beta_{e}\sigma_q^2  \norm{\vect{g}}^2 + 1.
\end{equation}

To obtain a lower bound on achievable rate of legitimate user $k$, we may express \eqref{eq7} as  a sum of signal and uncorrelated noise \cite{Marzetta2014} \cite{T2018}, i.e., 
\begin{equation} \label{eq9}
r_k  = a s_k + n_{\text{eff}}
\end{equation}
where $a$ is a deterministic constant which depends only on the statistics of the channel and $n_{\text{eff}}$ is an effective noise uncorrelated with $s_k$. 

From \eqref{eq9} the variance of  $n_{\text{eff}}$  is given by
\begin{align} \label{var_effective_noise}
\sigma_{n_{\text{eff}}}^2  = E\left [ (r_k  - a s_k )(r_k  - a s_k )^{\ast}\right] .
\end{align}
Thus, $\sigma_{n_{\text{eff}}}^2$ is minimized by choosing $a$ according to the optimal estimator in the sense of MMSE which renders $n_{\text{eff}}$ uncorrelated with the signal $s_k$. In the view of orthogonality property  of optimal estimator, i.e., $E \left [(r_k  - a s_k )s_k^{\ast}\right ]=0$, it follows that 
\begin{equation} \label{eq10}
a = E[s_k^{\ast} r_k] = c_1 \sqrt{\beta_{k}} E[ \vect{h}_k^T  \vect{w}_k]
\end{equation}

Using this result in \eqref{var_effective_noise} yields 
\begin{align} \label{eq10.1}\nonumber
&\sigma_{n_{\text{eff}}}^2=E[|r_k|^2]-|a|^2\\\nonumber
&=
c_1^2 \beta_k \operatorname{Var} (\vect{h}_k^T  \vect{w}_k)+ \sum \limits_{j=1, j\neq k}^K  c_1^2 \beta_k  E[ |\vect{h}_k^T  \vect{w}_j|^2] \\
& + c_2^2 \beta_k E[ \vect{h}_k^T \vect{S} \vect{h}_k^{\ast}] + \beta_k \sigma_q^2 p_d   +1.
\end{align}
where $\operatorname{Var}(\cdot)$ is the variance operator. In \eqref{eq_cov_eavesdropper} and \eqref{eq10.1} we have used the fact that $\vect{SS}^H = \vect{S}$ for both R-AN and NS-AN schemes.

Finally, a lower bound on the achievable rate $\underline{R}_k$ is obtained by treating the non-Gaussian noise $n_{\text{eff}}$ (which is uncorrelated with signal) as independent Gaussian noise with the same variance $\sigma_{n_{\text{eff}}}^2$ \cite{Hissibi2003}. Thus, we have
\begin{equation}\label{eq10.22}
\underline{R}_k= \log\left(1+\frac{|a|^2}{\sigma_{n_{\text{eff}}}^2}\right).
\end{equation}

For the convenience of exposition and analysis in this paper, we summarize the results for the eavesdropper and legitimate user in Lemmas \ref{R_e_general} and \ref{R_k_general} which will be used later in Sec. \ref{sec3.4}. 
\begin{lemma}
	\label{R_e_general}
	An upper bound on the ergodic information rate (leakage) of the  eavesdropper is given by 
	\begin{equation}
	 \label{eq:R_e_general}
\overline{R}_{e}=  E \left[ \log \left ( 1 +  c_1^2 \beta_{e} \frac{ \norm{\vect{w}_k^H \vect{g}^{\ast}}^2 }{\sigma_{e}^{2}}\right) \right]
	\end{equation}
	where $\sigma_{e}^{2}  =  c_2^2 \beta_{e} \vect{g}^T \vect{S} \vect{g}^{\ast}+ c_3^2 \beta_{e}\sigma_q^2  \norm{\vect{g}}^2 + 1$.
	\end{lemma}

	\begin{lemma} 
	\label{R_k_general}
	A lower bound on achievable rate of the legitimate user $k$ (intercepted) is given by 
	\begin{equation}\label{eq10.2}
	\underline{R}_k= \log\left(1+\frac{|a|^2}{\sigma_{n_{\text{eff}}}^2}\right)
	\end{equation}
	where $a$ and $\sigma_{n_{\text{eff}}}^2$ are given by
	\begin{align}
	a &= c_1 \sqrt{\beta_{k}} E[ \vect{h}_k^T  \vect{w}_k] \\ \nonumber
	\sigma_{n_{\text{eff}}}^2&= c_1^2 \beta_k (\operatorname{Var} (\vect{h}_k^T  \vect{w}_k)+ \sum \limits_{j=1, j\neq k}^K   E[ |\vect{h}_k^T  \vect{w}_j|^2]) \\
	& + c_2^2 \beta_k E[ \vect{h}_k^T \vect{S} \vect{h}_k^{\ast}] + \beta_k \sigma_q^2 p_d   +1.
	\end{align}
\end{lemma}

%==========================================================
\subsection {Impact of pilot attack}
\label{sec3.3}
%==========================================================

Because of the pilot attack, the estimated channel of the legitimate user $k$ will contain information (i.e., correlation) about the channel of the eavesdropper. Here, we characterize this information which turns to be useful in our analysis of the main results. 
%===================
\begin{lemma} \label{lemma1}
	The eavesdropper's channel vector  can be expressed as
	\begin{equation}\label{lemma1_eq1}
	\vect{g} = \sqrt{\kappa_R} \hat{ \vect{h}}_k + \vect{\epsilon}
	\end{equation}
	where $\kappa_R$ is the received power ratio between the eavesdropper and  intercepted user $k$, i.e., 
	\begin{equation} \label{kappa_R}
	\kappa_R  = \frac{ p_{e}^{\prime}}{ p_k^{\prime}} = \frac{\beta_e p_{e}}{\beta_k p_k}
	\end{equation}
	and $\vect{\epsilon}$ is uncorrelated Gaussian (approximately) error vector with covariance matrix given by 
	\begin{align}\label{lemma1_eq2}
	\vect{C}_{\vect{\epsilon}} = (1-{\kappa_R} \sigma_{\hat{h}_{k}}^2 ) \vect{I}_N 
	\end{align} 
\end{lemma}
\begin{proof}
The proof is straightforward which follows from the classical work on MMSE solution. Appendix \ref{Proof_lemma1} presents the details.
\end{proof}

Although Lemma \ref{lemma1} is a straightforward result, however, it is noteworthy. It can tell us how much information about the eavesdropper's channel $\vect{g}\sim  \mathcal{CN}(\vect{0}, \vect{I}_N)$ is contained in the channel estimate $\hat{ \vect{h}}_k\sim \mathcal{CN}(\vect{0}, \sigma_{\hat{h}_{k}}^2\vect{I}_N)$. 

\noindent Using Lemma \ref{lemma1}, the mutual information between eavesdropper's and legitimate user's channels is obtained as follows:
\begin{align} 
\label{eq_mutualinfo}
\nonumber
I(\vect{g}; \hat{\vect{h}}_k) &= h(\vect{g})-h(\vect{g}|\hat{\vect{h}}_k) \\ \nonumber
&=  h(\vect{g})-h(\vect{\epsilon}) \\\nonumber
&=N \log(\pi e)-N \log\left(\pi e (1-{\kappa_R} \sigma_{\hat{h}_{k}}^2)\right)\\
&=N \log\left( \frac{1}{1-\kappa_R \sigma_{\hat{h}_{k}}^2} \right) \ge 0, 
\end{align}

The intuitive result in \eqref{eq_mutualinfo} indicates that $I(\vect{g};\hat{\vect{h}}_k)$ can grow large and will be limited only by AWGN and quantization noise when $\kappa_R\gg1$. The equality in \eqref{eq_mutualinfo} is satisfied when $\kappa_R = 0$, i.e., passive eavesdropping. Note that when $\kappa_R\gg 1$ (eavesdropper's received power is much larger than received power of legitimate user), $ \sigma_{\hat{h}_{k}}^2$ becomes very small (i.e., channel estimate becomes unreliable) and vice versa. However, the product $\kappa_R \sigma_{\hat{h}_{k}}^2$ is always less than unity.

Since in particular the nullspace noise is a function of $\hat{ \vect{h}}_k$, which is correlated with $\vect{g}$, part of this noise lives in the nullspace of the eavesdropper's channel. Thus this part of nullspace noise will be annihilated at the eavesdropper, giving rise to an increase in his information rate and hence a significant reduction in secrecy rate.

%==========================================================
\subsection{Main theoretical results}
\label{sec3.4}
%==========================================================

Here, we give a lower bound on the achievable secrecy rate under different beamforming and artificial noise techniques. In the following, all derived information rates are given in their normalized form\footnote{The normalization factor is ($1-\tau/T_c$), i.e., the fraction of time over which downlink transmission is considered in this work.}.

We state our findings in the following two theorems.   
%=======main result 1=========================
\begin{thm}
	\label{thm1}
	Consider a one-bit quantized Massive MIMO system with $N$ antennas at the BS and $K$ single-antenna users in the presence of a single-antenna active eavesdropper. Also, imperfect CSI is assumed to be available at the BS. If BS uses MRT-BF, then a lower bound on the achievable rate of the intercepted user $k$,  is given by
	\begin{equation} \label{thm1_MRT}
	\underline{R}_{k}^{\textnormal{MRT}} = \log\left (1+\frac{2 \theta \pi^{-1} \operatorname{tr}^{-1}(\vect{\Sigma}) \beta_k \sigma_{\hat{h}_{k}}^4  p_d N}{{2  \theta\beta_k p_d}/{\pi} +P_{k}^{\textnormal{AN}}+\beta_k \sigma_q^2 p_d+1}\right).
	\end{equation}
Further, if the BS uses ZF-BF, then a lower bound on the achievable rate is given by 
\begin{equation}\label{thm1_ZF}
	\underline{R}_{k}^{\textnormal{ZF}} =  \log \left( 1+ \frac{2 \theta  \pi^{-1} \operatorname{tr}^{-1}(\vect{\Sigma}^{-1}) \beta_k p_d (N-K)}{{2 \theta \beta_k p_d }(1-\sigma_{\hat{h}_{k}}^2) /\pi +P_{k}^{\textnormal{AN}} + \beta_k \sigma_q^2 p_d +1}\right)
	\end{equation}
	where  $P^{\textnormal{AN}} $ is the leakage power of artificial noise seen at the intercepted user $k$ defined as
	\begin{equation} \label{thm1_AN}
	P_{k}^{\textnormal{AN}} =
	\begin{cases}
	{2 \bar{\theta} \beta_k p_d }/{\pi}&\textnormal{ R-AN }  \\
	{2 \bar{\theta} \beta_k p_d }(1-\sigma_{\hat{h}_{k}}^2)/\pi &\textnormal{NS-AN}.
	\end{cases}
	\end{equation} 
\end{thm}
\begin{proof}
	See Appendix \ref{proof_thm1}.
\end{proof}

%=======main result 2=========================
\begin{thm}
	\label{thm2}
	Consider the system model in Theorem \ref{thm1}. When the number of base station antennas $N$ is sufficiently large, an upper bound on the ergodic information rate leaked to the eavesdropper is given by (equal or approximate)
	
	\begin{equation} \label{thm2_eq1}
	\overline{R}_{e}^{\textnormal{MRT}} \cong    \log \left( 1+ \frac{ 2 \theta   \beta_{e} p_d \sigma_{\hat{h}_{k}}^2  \left (\kappa_R \sigma_{\hat{h}_{k}}^2 N+ 1 \right) } { \pi \operatorname{tr}(\vect{\Sigma}) (P_{e}^{\textnormal{AN}}   + \beta_{e} p_d \sigma_q^2 + 1)}\right)
	\end{equation}
	when the BS uses MRT-BF, and when the BS uses ZF-BF, 
	\begin{equation}\label{thm2_eq2}
	\overline{R}_{e}^{\textnormal{ZF}} \cong   \log \left( 1+ \frac{ 2 \theta  \beta_{e} p_d \left (\kappa_R (N-K-1) + \sigma_{\hat{h}_{k}}^{-2}    \right) } { \pi \operatorname{tr}(\vect{\Sigma}^{-1}) (P_{e}^{\textnormal{AN}}   + \beta_{e} p_d \sigma_q^2 + 1)}\right)
	\end{equation}
where $\kappa_R$ is the receive power ratio defined in \eqref{kappa_R} and $P_{e}^{\text{AN}} $ is the leakage power of artificial noise seen at the eavesdropper defined as
	\begin{equation}\label{thm2_AN}
	P_{e}^{\textnormal{AN}} =
	\begin{cases}
	2\bar{\theta}  \beta_{e} p_d/\pi  &\textnormal{R-AN }  \\
	{2 \bar{\theta} \beta_{e} p_d }(1-\kappa_R  \sigma_{\hat{h}_{k}}^2)/\pi &\textnormal{NS-AN}.
	\end{cases}
	\end{equation} 
\end{thm}
\begin{proof}
	See Appendix \ref{proof_thm2}
\end{proof}

From \eqref{thm1_MRT} and \eqref{thm1_ZF}  we identify the different components of noise at the legitimate user $k$ as follows. The term ${2  \theta\beta_k p_d}/{\pi}$ in \eqref{thm1_MRT} or  ${2 \theta \beta_k p_d }(1-\sigma_{\hat{h}_{k}}^2) /\pi$ in \eqref{thm1_ZF} captures the effect of \emph{beamforming gain penalty \footnote{Beamforming gain penalty is due to CSI uncertainty at the user since user relies on channel statistics rather than instantaneous channel realization.}plus inter-user interference}, the term $P_{k}^{\text{AN}}$ captures the leakage power of \emph{artificial noise}, and the term $\beta_k \sigma_q^2 p_d+1$ captures the effect of \emph{quantization noise and AWGN}. For more details, see Appendix \ref{proof_thm1}.

It is clear from \eqref{thm1_AN} and \eqref{thm2_AN} that both the legitimate user and eavesdropper achieve higher data rates when the BS employs NS-AN than R-AN. Furthermore, their information rates increase with increasing the number of BS antennas $N$ and decrease with increasing the number of users $K$. Note that the dependence of the rates on $K$ when MRT-BF is used is captured by $\operatorname{tr}(\vect{\Sigma})$, whereas  captured by the factors $N-K$ and $\operatorname{tr}(\vect{\Sigma}^{-1})$ when ZF-BF is used. 

Also, by inspecting  \eqref{thm1_AN}, it is obvious that when perfect CSI is available at the BS, the nullspace noise seen by the legitimate user becomes 0 due to $\sigma_{\hat{h}_{k}}^2=1$. That is to say, the artificial noise is perfectly aligned with the nullspace of the channel.  Again, from \eqref{thm2_AN}, the assumption of perfect CSI implies $\kappa_R=0$ (passive eavesdropper), thus rendering both the random and nullspace artificial noises have the same negative effect on the information rate from the eavesdropper’s perspective. 

%==========================================================
\subsection{Achievable secrecy rate}
\label{sec3.5}
%==========================================================
From \eqref{thm1_ZF} and \eqref{thm2_eq2}, we rewrite the respective rates  $\underline{R}_{k}^{\textnormal{ZF}}$ and $\overline{R}_{e}^{\textnormal{ZF}}$ as follows:
  \begin{equation}
\label{thm1_ZF11}
\underline{R}_{k}^{\textnormal{ZF}} 
= \log \left( 1+ \frac{2 \theta  \beta_k p_d (N-K)}{A_1}\right)
\end{equation}
 \begin{IEEEeqnarray}{rcl}
\label{thm2_ZF11}
\overline{R}_{e}^{\textnormal{ZF}} \cong  \log \left( 1+ \frac{ 2 \theta  \beta_{e} p_d \ (\kappa_R (N-K-1) + \sigma_{\hat{h}_{k}}^{-2} ) } { A_2}\right) 
\end{IEEEeqnarray}
where $A_1,A_2$ are defined as
\begin{IEEEeqnarray}{rcl}
\nonumber
A_1 &=& \pi \operatorname{tr}(\vect{\Sigma}^{-1}) ( P_{k}^{\text{AN}}  +\frac {2 \theta \beta_k p_d }{\pi}(1-\sigma_{\hat{h}_{k}}^2)  + \beta_k \sigma_q^2 p_d +1) \\
A_2 &=& \pi \operatorname{tr}(\vect{\Sigma}^{-1}) (P_{e}^{\textnormal{AN}}   + \beta_{e} p_d \sigma_q^2 + 1).
\end{IEEEeqnarray}

Using \eqref{eq:Rs_Def}, we define $\underline{R}_s^{\text{ZF}} (p_d, \theta, N) =\left [\underline{R}_{k}^{\text{ZF}} - \overline{R}_{e}^{\text{ZF}}\right ]^+$ as a lower bound on secrecy rate\footnote{The explicit use of the parameters $p_d, \theta$ and $N$  in the secrecy rate is  introduced only for reasons of mathematical convenience.} when ZF-BF is used. Thus
\begin{align} 
\label{eq_thm12_sec_rate2}
\nonumber
\underline{R}_s^{\text{ZF}} (p_d, \theta, N) &= \Bigg[ \log \left( 1+ \frac{2 \theta  \beta_k p_d (N-K)}{A_1}\right)-\log \bigg( 1+  \\ \nonumber
&\frac{ 2 \theta  \beta_{e} p_d  (\kappa_R (N-K-1) + \sigma_{\hat{h}_{k}}^{-2} ) } { A_2}\bigg) \Bigg]^{+}\\
& = \left[\log\left(\frac{2 A_2 \theta  \sigma_{\hat{h}_{k}} ^2 p_d \beta _k N + C_1}{2 A_1 \theta  \kappa_R   \sigma_{\hat{h}_{k}} ^2 p_d \beta _e N  +C_2}\right)\right]^{+}
\end{align}

%\begin{align} 
%\label{eq_thm12_sec_rate2}
%\nonumber
%\underline{R}_s^{\text{ZF}} (p_d, \theta, N)&=\left [\underline{R}_{k}^{\text{ZF}} - \overline{R}_{e}^{\text{ZF}}\right ]^+ \\ 
%& = \left[\log\left(\frac{2 A_2 \theta  \sigma_{\hat{h}_{k}} ^2 p_d \beta _k N  + C_1}{2 A_1 \theta  \kappa   \sigma_{\hat{h}_{k}} ^2 p_d \beta _e N  +C_2}\right) \right]^+
%\end{align}
where $C_1$ and $C_2$ are defined as
%-----------------------------------------------------------------------
\begin{equation}
\begin{aligned}
C_1 &= A_2  (A_1-2 \theta  K p_d \beta _k)\sigma_{\hat{h}_{k}} ^2 \\
C_2 &=A_1  (A_2 \sigma_{\hat{h}_{k}} ^2-2 \theta  p_d \beta _e (\kappa_R  (K+1) \sigma_{\hat{h}_{k}} ^2-1))
\end{aligned}
\end{equation}
%-----------------------------------------------------------------------
Likewise, from \eqref{thm1_MRT} and \eqref{thm2_eq1}, we write the respective rates  $\underline{R}_{k}^{\textnormal{MRT}}$ and $\overline{R}_{e}^{\textnormal{MRT}}$ as follows: 
\begin{equation} \label{thm1_MRT11}
\underline{R}_{k}^{\textnormal{MRT}} =\log\left (1+\frac{2 \theta \beta_k \sigma_{\hat{h}_{k}}^4  p_d N}{B_1}\right)
\end{equation}

\begin{equation} \label{thm2_MRT11}
\overline{R}_{e}^{\textnormal{MRT}} \cong  \log \left( 1+ \frac{ 2 \theta   \beta_{e} p_d \sigma_{\hat{h}_{k}}^2   (\kappa_R \sigma_{\hat{h}_{k}}^2 N+ 1 ) } { B_2}\right)
\end{equation}
where $B_1$ and $B_2$ are defined as
\begin{equation}
\begin{aligned}
B_1 &=\pi \operatorname{tr}(\vect{\Sigma}) ({2  \theta\beta_k p_d}/{\pi} +P_{k}^{\text{AN}}+\beta_k \sigma_q^2 p_d+1) \\
B_2 &= \pi \operatorname{tr}(\vect{\Sigma}) (P_{e}^{\textnormal{AN}}   + \beta_{e} p_d \sigma_q^2 + 1).
\end{aligned}
\end{equation}

Denoting $\underline{R}_s^{\text{MRT}}(p_d, \theta, N)=  \left[\underline{R}_{k}^{\text{MRT}} - \overline{R}_{e}^{\text{MRT}}  \right]^+$ as a lower bound on secrecy rate when MRT-BF is employed, thus we write 
\begin{align} \label{eq_thm12_sec_rate1} \nonumber
\underline{R}_s^{\text{MRT}}(p_d, \theta, N)&=\Bigg[ \log\left (1+\frac{2 \theta \beta_k \sigma_{\hat{h}_{k}}^4  p_d N}{B_1}\right)-\log \bigg( 1+  \\\nonumber
&\frac{ 2 \theta   \beta_{e} p_d \sigma_{\hat{h}_{k}}^2   (\kappa_R \sigma_{\hat{h}_{k}}^2 N+ 1 ) } { B_2}\bigg) \Bigg]^{+} \\
&=\left[ \log\left(\frac{2 B_2 \theta   \sigma_{\hat{h}_{k}} ^4 p_d \beta _k N  + C_3}{2 B_1 \theta  \kappa_R   \sigma_{\hat{h}_{k}} ^4 p_d \beta _e N + C_4}\right)\right]^{+}
\end{align}

%\begin{align}
%\nonumber
%\underline{R}_s^{\text{MRT}}(p_d, \theta, N)&:= \underline{R}_s^{\text{MRT}} =  \left[\underline{R}_{k}^{\text{MRT}} - \overline{R}_{e}^{\text{MRT}}  \right]^+ \\ 
%& = \Bigg[ \log\left(\frac{2 B_2 \theta   \sigma_{\hat{h}_{k}} ^4 p_d \beta _k N  + C_3}{2 B_1 \theta  \kappa_R   \sigma_{\hat{h}_{k}} ^4 p_d \beta _e N + C_4}\right) \Bigg]^+
%\end{align}
where $C_3$ and $C_4$ are  defined as
%-----------------------------------------------------------------------
\begin{equation}
\begin{aligned}
C_3 &=B_1 B_2\\
C_4&= B_1 B_2 + 2 B_1 \theta  \sigma_{\hat{h}_{k}} ^2 p_d \beta _e.
\end{aligned}
\end{equation}

Because of the concavity of the $\log(\cdot)$ and the non-monotonic behavior of  $\underline{R}_{k}^{\text{ZF}} - \overline{R}_{e}^{\text{ZF}}$ and $\underline{R}_{k}^{\text{MRT}} - \overline{R}_{e}^{\text{MRT}}$ with respect to $\theta$, the lower bounds \eqref{eq_thm12_sec_rate2} and \eqref{eq_thm12_sec_rate1} can be maximized with respect to $\theta$. The optimal $\theta$ maximizing the secrecy rates $\underline{R}_s^{\text{MRT}}(p_d, \theta, N)$ and $\underline{R}_s^{\text{ZF}}(p_d, \theta, N)$ can be found by solving ${d \underline{R}_s^{\text{MRT}}(p_d, \theta, N)}/d\theta = 0$ and ${d \underline{R}_s^{\text{ZF}}(p_d, \theta, N)}/d\theta = 0$ for $\theta \in (0,1)$. 

Due to the cumbersome algebraic expressions of optimal $\theta$, we omit them and hence compute the optimal values numerically instead. As shown in the next section, it turns out that when all parameters are fixed, the optimal policy is to allocate almost all power to artificial noise (i.e., signal power becomes infinitesimal) in the asymptotic limit of $N$.

%==========================================================
\subsection{Secrecy rate under passive eavesdropping }
\label{sec3.6}
%==========================================================

Passive eavesdropping corresponds to the situation where the eavesdropper does not transmit any signal (i.e., $p_e = 0, \kappa_R=0$) to conceal himself and his harm is limited only by eavesdropping on the downlink transmission to decode the confidential message sent to the legitimate user. Based on the CSI's availability at the BS, we study the following two scenarios. 

\textit{Passive eavesdropping with imperfect CSI} (P-ICSI):
Perfect CSI is rarely available at the BS, and hence it needs to be estimated beforehand. To specialize our results of the quantized system to this scenario, we redefine the variance of channel estimate $\sigma_{\hat{h}_{l}}^2$ and the diagonal matrix $\vect{\Sigma}$ derived previously as follows:
\begin{equation} 
\label{eq:Variable_change1}
\begin{cases}
 \sigma_{\hat{h}_{l}}^2 &\xrightarrow{p_e =0} \xi_l^2 =  \frac{\gamma^2 p_{l}^{\prime} \tau}{\gamma^2   p_{l}^{\prime} \tau +  \gamma^2 +   \sigma_q^2 }\\
 \vect{\Sigma} &\xrightarrow{p_e =0}  \vect{ \Pi} =  \operatorname{diag}(\xi_1^2, \xi_2^2, \cdots,\xi_K^2).
 \end{cases}
\end{equation} 

Substituting \eqref{eq:Variable_change1} in Theorem \ref{thm1}, \eqref{eq:Variable_change1} with $\kappa_R=0$ in Theorem \ref{thm2}, and using the definition \eqref{eq:Rs_Def}, we obtain the achievable secrecy rate ${\underline{R}_s^{\text{P-ICSI}}}$ shown in \eqref{eq:P-ICSI} (on the top of this page), where $P_{k, \textnormal{P-ICSI} }^{\textnormal{AN}} $ is defined as
\begin{figure*}[!btp]
	\normalsize
	\begin{equation}
	{\underline{R}_s^{\text{P-ICSI}} =}
	\begin{cases} 
	\left [   \log \left(1 + \frac{2 \operatorname{tr^{-1}}(\mathbf{\Pi} )  \xi_k^4 \theta  p_d \beta _k N }{ p_d \beta _k \left(2 \theta +\pi  \sigma _q^2\right)+\pi  (P_{k, \textnormal{P-ICSI} }^{\textnormal{AN}}+1)}\right)-\log \left(1+ \frac{2 \operatorname{tr^{-1}}(\mathbf{\Pi} )  \xi_k^2 \theta  p_d \beta _e}{p_d \beta _e \left(2-2 \theta +\pi  \sigma _q^2\right)+\pi}\right) \right]^+ & \text{MRT-BF} \\ 
	\left[ \log \left(1+\frac{2 \operatorname{tr^{-1}}(\mathbf{\Pi^{-1}} )  \theta  p_d \beta _k (N-K)}{ p_d \beta _k \left(2 \theta -2 \xi_k^2 \theta +\pi  \sigma_q^2\right)+\pi  (P_{k, \textnormal{P-ICSI} }^{\textnormal{AN}}+1) }\right) - \log \left(1 + \frac{2 \operatorname{tr^{-1}}(\mathbf{\Pi^{-1}} ) \xi_k^{-2}  \theta  p_d \beta _e}{p_d \beta _e \left(2-2 \theta +\pi  \sigma _q^2\right)+\pi }\right) \right]^+&\text{ZF-BF}  
	\end{cases}
	\label{eq:P-ICSI}
	\end{equation}
	\hrulefill
	\vspace*{0pt}
\end{figure*}

\begin{align} \label{eq:P-ICSI_AN}
P_{k, \textnormal{P-ICSI} }^{\textnormal{AN}} &=
\begin{cases}
\frac{2 \bar{\theta} p_d \beta _k}{\pi }&\text{R-AN }  \\
\frac{2 \bar{\theta} p_d \beta _k \left(1-\xi_k^2\right)  }{\pi }&\text{NS-AN}.
\end{cases} 
\end{align}

\textit{Passive eavesdropping with perfect CSI} (P-PCSI): The assumption of perfect CSI at the BS corresponds to $\xi_l = 1$ and hence $\vect{\Pi} = \vect{I}_K$. Thus, substituting $\xi_k = 1$ and $\vect{\Pi} = \vect{I}_K$ in \eqref{eq:P-ICSI} and \eqref{eq:P-ICSI_AN}, the secrecy rate in \eqref{eq:P-PCSI} (on the top of this page) follows directly, where
 \begin{figure*}[!btp]
	\normalsize
\begin{equation}
{\underline{R}_s^{\text{P-PCSI}} =}
\begin{cases} 
\left [   \log \left(1+ \frac{2 \theta  p_d \beta _k N/K}{ p_d \beta _k \left(2 \theta +\pi  \sigma _q^2\right)+\pi  (P_{k, \textnormal{P-PCSI} }^{\textnormal{AN}} +1)}\right)- \log \left(1+\frac{2 \theta  p_d \beta _e/K }{p_d \beta _e \left(2-2 \theta +\pi  \sigma _q^2\right)+\pi}\right) \right]^+ & \text{MRT-BF} \\ 
\left[\underbrace{\log \left(1+\frac{2 \pi^{-1} \theta  p_d \beta _k (N/K-1)}{p_d \beta _k \sigma_q^2 +P_{k, \textnormal{P-PCSI} }^{\textnormal{AN}} +1 }\right)}_{R_k^{\text{P-PCSI}}} -  \underbrace{\log \left(1+\frac{2 \theta  p_d \beta _e /K}{p_d \beta _e \left(2-2 \theta +\pi  \sigma _q^2\right)+\pi}\right)}_{R_e^{\text{P-PCSI}}} \right]^+&\text{ZF-BF}  
\end{cases}
\label{eq:P-PCSI}
\end{equation}
\hrulefill
\vspace*{0pt}
\end{figure*}
\begin{equation} 
P_{k, \textnormal{P-PCSI} }^{\textnormal{AN}} =
\begin{cases}
{2 \bar{\theta}p_d \beta _k}/{\pi } &\text{R-AN }  \\
0&\text{NS-AN}.
\end{cases}
\end{equation} 
The result of ZF-BF in \eqref{eq:P-PCSI} can be directly obtained from [\cite{8626548},eq. (27)]\footnote{More specifically, $R_k^{\text{P-PCSI}}$ in \eqref{eq:P-PCSI} can be directly obtained from [\cite{8626548}, Eqs. (19) and (20)] and $R_e^{\text{P-PCSI}}$ in \eqref{eq:P-PCSI} is obtained from [\cite{8626548}, Eq.(23)] with the assumption that the variance of AWGN at the eavesdropper is  0 and the ratio $\alpha = K/N$ is set to 0 instead of being fixed as assumed in \cite{8626548}.}.

Finally, we remark that under passive eavesdropping, both R-AN and NS-AN look like random noise from the perspective of eavesdropper and hence they have the same effect on his rate. Further, while the achievable rate of legitimate user increases with increasing $N$, however, the eavesdropper's rate is independent of the number of BS antenna. Therefore, the secrecy rate steadily increases with $N$, contrary to the case of active eavesdropping with a single-antenna eavesdropper. This observation has already been reported in the literature (i.e., see \cite{Kapetanovic2015}) under unquantized systems and hence it holds true for quantized systems as well.

%==========================================================
\subsection{Secrecy rate in unquantized system}
\label{sec3.7}
%==========================================================
With the unquantized system, it is assumed that the BS has access to the original received signal (uplink) and the transmitted signal (in downlink) undergoes no quantization, i.e.,  the BS is assumed to have infinite-resolution ADCs and DACs. 

To specialize our results of the quantized system to the unquantized system, we redefine the set of parameters $\{p_d, \sigma_q^2 ,\sigma_{\hat{h}_{l}}^2,  \vect{\Sigma} \}$ as follows: 
\begin{equation} 
\label{eq:Variable_change2}
\begin{cases}
p_d &\rightarrow {p_d \pi}/{2} ,  \sigma_q^2  \rightarrow  0 \hspace{5pt} (\textnormal{Remark} \ref{remark:rem_UQ})\\
  \sigma_{\hat{h}_{l}}^2 &\xrightarrow{\sigma_q^2=0, \gamma=1}   \vartheta_l^2  = \frac{ p_{l}^{\prime} \tau}{  p_{l}^{\prime} \tau +  p_{e}^{\prime} \tau \delta(l-k)+  1  }\\
 \vect{\Sigma} &\rightarrow \vect{ \Xi} =  \operatorname{diag}(\vartheta_1^2,\vartheta_2^2,\cdots,\vartheta_K^2).
\end{cases}
\end{equation} 

Considering the parameters' replacement \eqref{eq:Variable_change2} in Theorems \ref{thm1} and \ref{thm2} and using the definition \eqref{eq:Rs_Def} yields the achievable secrecy rate as given in \eqref{eq:Rs_UQ} shown at the top of page \pageref{eq:Rs_UQ}, where  $P_{k, \textnormal{UQ}}^{\textnormal{AN}} $ and $P_{e, \textnormal{UQ}}^{\textnormal{AN}}$ are, respectively, given by
\begin{equation} 
\label{thm1_AN2222}
P_{k, \textnormal{UQ}}^{\textnormal{AN}} =
\begin{cases}
\bar{\theta}  \beta _k p_d &\text{R-AN }  \\
\bar{\theta}  \beta _k p_d  \left(1-\vartheta_k^2 \right)   &\text{NS-AN}
\end{cases}
\end{equation} 
\begin{equation}
\label{thm2_AN222}
P_{e, \textnormal{UQ}}^{\textnormal{AN}} =
\begin{cases}
\bar{\theta}  \beta _e p_d&\textnormal{R-AN }  \\
\bar{\theta}  \beta _e p_d \left(1-\kappa_R\vartheta_k^2   \right) &\textnormal{NS-AN}.
\end{cases}
\end{equation}

%==========================================================
\section{Asymptotic Performance comparison}
\label{sec4}
%==========================================================
Inspecting the  secrecy rates in \eqref{eq_thm12_sec_rate2} and \eqref{eq_thm12_sec_rate1} provides no clear clue of how the performance of MRT-BF and ZF-BF can be compared. Therefore, a better understanding of the performance gap can be gained through asymptotic performance. Our focus here will be on the asymptotic behavior of the beamforming/artificial noise schemes as the number of BS antennas increases with no limit. As shown next, the asymptotic performance renders it easy to capture the important parameters for a specific scheme to guarantee a positive secrecy rate, which turns to be even very useful for the non-asymptotic case.

In Massive MIMO the transmit power of the BS can be cut down as the number of BS antennas grows large while maintaining a nonzero data rate for each user in the system, i.e.,  \textit{power-scaling law} \cite{6457363}. Since our concern is the secrecy rate rather than the conventional rate, thus,  it is of interest to know whether the power-scaling law remains valid. That is to say, we seek to see if it is possible to reduce the transmit power at the BS as $N \to \infty$ while maintaining a nonzero secrecy rate.

In the following, we study the asymptotic behavior of the secrecy rate in Massive MIMO system with and without transmit power scaling at the BS. 

%==========================================================
\subsection{Massive MIMO with no power scaling}
\label{sec4.1}
%==========================================================

Here, we assume that the transmit power $p_d $ at the BS is not scaled down as $N$ grows large (independent of $N$), i.e., no power scaling (no-PS). In the following, we state our results in the following corollary.
\begin{cor} \label{cor1}
	Assume the BS uses MRT-BF or ZF-BF. Then when R-AN is used, the maximum secrecy rate converges to
	\begin{equation} \label{eq_cor1}
	\underline{R}_{s, \textnormal{R-AN}}^{\textnormal{no-PS}} \rightarrow\left[ \log \left( \frac{\beta_k  (p_d  \beta _e (\pi  \sigma _q^2+2)+\pi ) }{  \kappa_R \beta _e   (p_d \beta _k  (\pi  \sigma _q^2+2)+\pi  )}\right)\right]^{+}
	\end{equation}
	and when NS-AN is used, the maximum secrecy rate converges to
	\begin{equation}\label{eq_cor2}
	\underline{R}_{s, \textnormal{NS-AN}}^{\textnormal{no-PS}} \rightarrow \left[\log\left( \frac{\beta _k (p_d \beta _e (\pi  \sigma _q^2+2-2 \kappa_R  \sigma_{\hat{h}_{k}}^2)+\pi )}{\kappa_R  \beta _e (p_d \beta _k (\pi  \sigma _q^2+2-2 \sigma_{\hat{h}_{k}}^2 )+\pi )} \right)\right]^{+}
	\end{equation}
	asymptotically as $N\to \infty$.
\end{cor}

\begin{proof}
	In the following we need to evaluate the secrecy rate $\underline{R}_s^{\text{ZF}}(p_d, \theta, N)$ \eqref{eq_thm12_sec_rate2} and $\underline{R}_s^{\text{MRT}}(p_d, \theta, N)$ \eqref{eq_thm12_sec_rate1} as $N\to \infty$, then maximize the resulting expressions with respect to $\theta$.
	
Consider the ZF-BF scheme. Taking the limit of \eqref{eq_thm12_sec_rate2} as $N\to \infty$ and using the definition of the leakage power of artificial noise \eqref{thm1_AN} and \eqref{thm2_AN}, we obtain
	\begin{IEEEeqnarray}{rCl} 
	\label{eq: cor1_proof1}
	\nonumber
&&f_1(\theta):= \lim_{N\to \infty}\underline{R}_s^{\text{ZF}}(p_d, \theta, N) \\ \nonumber
&=& \left[\lim_{N\to \infty} \log\left(\frac{2 A_2 \theta  \sigma_{\hat{h}_{k}} ^2 p_d \beta _k N  + C_1}{2 A_1 \theta  \kappa _R \sigma_{\hat{h}_{k}} ^2 p_d \beta _e  N  +C_2}\right) \right]^+ \\ \nonumber
 &=&\left[\log\left( \frac{A_2 \beta _k}{A_1   \beta _e \kappa_R} \right ) \right]^+\\ \nonumber
&=&\left[\log \left ( \frac{\beta _k (P_{e}^{\textnormal{AN}}+p_d \beta _e \sigma _q^2+ 1)}{\kappa_R  \beta _e ( P_{k}^{\text{AN}}  + \frac{2 \theta \beta_k p_d }{\pi}(1-\sigma_{\hat{h}_{k}}^2) + \beta_k \sigma_q^2 p_d +1)} \right) \right]^+\\ 
&=& \begin{cases}
\left[\log \left (\frac{\beta _k \left(p_d \beta _e (2-2 \theta +\pi  \sigma _q^2)+\pi \right)}{\kappa_R  \beta _e (p_d \beta _k (2-2 \theta  \sigma_{\hat{h}_{k}} ^2+\pi  \sigma_q^2 )+\pi )} \right) \right]^+&\textnormal{R-AN }  \\ \nonumber
\left[\log \left( \frac{\beta _k (p_d \beta _e (2 (\theta -1) (\kappa_R  \sigma_{\hat{h}_{k}} ^2-1)+\pi  \sigma _q^2)+\pi )}{\kappa_R  \beta _e (p_d \beta _k (\pi  \sigma_q^2-2 \sigma_{\hat{h}_{k}}^2+2)+\pi )} \right) \right]^+&\textnormal{NS-AN}. 
\end{cases}
\\
\end{IEEEeqnarray}

When MRT-BF is used, we proceed as follows. Taking the limit of \eqref{eq_thm12_sec_rate1} as $N\to \infty$ and using the definition of the leakage power of artificial noise \eqref{thm1_AN} and \eqref{thm2_AN}, we obtain
\begin{align}
	\label{eq: cor1_proof2}
	\nonumber
	&f_2(\theta) :=\lim_{N\to \infty}\underline{R}_s^{\text{MRT}}(p_d, \theta, N) \\ \nonumber
	&=\left[\lim_{N\to \infty} \log\left(\frac{2 B_2 \theta   \sigma_{\hat{h}_{k}} ^4 p_d \beta _k N  + C_3}{2 B_1 \theta  \kappa_R   \sigma_{\hat{h}_{k}} ^4 p_d \beta _e N + C_4}\right) \right]^+\\ \nonumber
	&=\left[\log \left(\frac{B_2 \beta _k}{B_1 \kappa_R  \beta _e} \right)\right]^+ \\  \nonumber
	&=\left[ \log\left(   \frac{ \beta _k  (P_{e}^{\textnormal{AN}}   + \beta_{e} p_d \sigma_q^2 + 1) }{ \kappa_R  \beta _e ({2  \theta\beta_k p_d}/{\pi} +P_{k}^{\text{AN}}+\beta_k \sigma_q^2 p_d+1) }    \right) \right]^+\\  \nonumber
	&=\begin{cases}
		\left[\log \left (\frac{\beta _k \left(p_d \beta _e (-2 \theta +\pi  \sigma _q^2+2)+\pi \right)}{\kappa  \beta _e (p_d \beta _k \left(\pi  \sigma _q^2+2\right)+\pi )} \right)\right]^+ &\textnormal{R-AN }  \\
	\left[	\log \left( \frac{\beta _k (p_d \beta _e (2 (\theta -1) (\kappa_R  \sigma_{\hat{h}_{k}} ^2-1)+\pi  \sigma _q^2)+\pi )}{\kappa_R  \beta _e (p_d \beta _k (2 (\theta -1) \sigma_{\hat{h}_{k}} ^2+\pi  \sigma _q^2+2)+\pi )} \right)\right]^+&\textnormal{NS-AN}.
	\end{cases}
	\\
\end{align}

%----------
\begin{figure*}[!bpt]
	\normalsize
	\begin{equation}
	\underline{R}_s^{\text{UQ}} =
	\begin{cases} 
	\left[	\log \left(1+ \frac{\operatorname{tr^{-1}}(\vect{\Xi} )  \vartheta_{k}^4 \theta  \beta _k  p_d  N}{\theta  p_d\beta _k + P_{k, \textnormal{UQ}}^{\textnormal{AN}} +1}\right) - \log \left(1 + \frac{\vartheta_{k} ^2 \theta  p_d \beta _e \left(\vartheta _{k}^2 \kappa_R  N+1\right)}{\text{tr}(\vect{\Xi}  ) ( P_{e, \textnormal{UQ}}^{\textnormal{AN}}+1) }\right) \right]^+ & \text{MRT-BF} \\ 
	\left[\log \left(1+\frac{\theta  p_d \beta _k (N-K) \operatorname{tr^{-1}}(\vect{\Xi^{-1}} ) }{\theta  p_d \beta _k (1-\vartheta_{k}^2) +P_{k, \textnormal{UQ}}^{\textnormal{AN}} +1}\right) -\log \left(1+\frac{\theta  p_d \beta _e \left({\vartheta _{k}^{-2}}+\kappa_R (N-K-1)\right)}{\text{tr}\left({\vect{\Xi^{-1}} }\right) (P_{e, \textnormal{UQ}}^{\textnormal{AN}}  +1) }\right)  \right]^+&\text{ZF-BF}  
	\end{cases}
	\label{eq:Rs_UQ}
	\end{equation}
	\hrulefill
	\vspace*{0pt}
\end{figure*}
%----------
Next, we maximize \eqref{eq: cor1_proof1} and \eqref{eq: cor1_proof2} with respect to $\theta$. It is easy to verify that the first derivatives $\frac{d}{d \theta}  f_1(\theta)$ and  $\frac{d}{d \theta}  f_2(\theta)$ have no critical points in $\theta \in (0,1)$ for both R-AN and NS-AN schemes, and $f_1(\theta)$ and $f_2(\theta)$ are decreasing functions of $\theta$. Thus the value of $\theta$ maximizing $f_1(\theta)$ and $f_2(\theta)$ coincides with the left endpoint ($\theta \to  0$), i.e., power allocated to signal becomes infinitesimal. 

Substituting  $\theta = 0$ in \eqref{eq: cor1_proof1} and \eqref{eq: cor1_proof2}, we obtain
\begin{IEEEeqnarray}{rCl} 
	\nonumber
	&&\underline{R}_{s}^{\textnormal{no-PS}}  =f_1(0)  = f_2(0)  \\ \nonumber
	& =&
	\begin{cases}
		\left[\log \left( \frac{\beta_k  \left(p_d  \beta _e \left(\pi  \sigma _q^2+2\right)+\pi \right) }{  \kappa_R \beta _e   (p_d \beta _k  \left(\pi  \sigma _q^2+2\right)+\pi  )}\right) \right]^+&\textnormal{if R-AN }  \\
		\left[\log\left( \frac{\beta _k (p_d \beta _e (\pi  \sigma _q^2+2-2 \kappa_R  \sigma_{\hat{h}_{k}}^2)+\pi )}{\kappa_R  \beta _e (p_d \beta _k (\pi  \sigma _q^2+2-2 \sigma_{\hat{h}_{k}}^2 )+\pi )} \right)\right]^+&\textnormal{if NS-AN}
	\end{cases}
	\\
\end{IEEEeqnarray} 
which is independent of the beamforming scheme.
\end{proof}

From \eqref{eq_cor1} and \eqref{eq_cor2}, a positive secrecy rate is possible if the transmit power ratio (during channel training) between the eavesdropper and intercepted user satisfies
%-------------------------------------------
\begin{equation}\label{eq_cor1_and_2_alpha}
\kappa_T = \frac{p_e}{p_k}< 1+\underbrace{ \frac{\pi  \beta _k \left(\beta _k-\beta _e\right)}{\left(p_d \beta _k \left(\pi  \sigma _q^2+2\right)+\pi \right) \beta _e^2  }}_{\Delta \beta}.
\end{equation}
Further, since \eqref{eq_cor1} and \eqref{eq_cor2} are positive under the same condition \eqref{eq_cor1_and_2_alpha}, we have that
%--------
\begin{equation}\label{eq_cor1_2_difference1}
\Delta^{\textnormal{no-PS}}= \underline{R}_{s, \text{NS-AN}}^{\textnormal{no-PS}}-\underline{R}_{s, \textnormal{R-AN}}^{\textnormal{no-PS}}  >0.
\end{equation}

%----------------------Remark---------------
We summarize our conclusions from Corollary \ref{cor1} as follows:
\begin{enumerate}
	\item  The NS-AN outperforms R-AN asymptotically, independent of the beamforming technique.
	\item Using R-AN entails more BS antennas to achieve the same performance of  NS-AN.
	\item Both NS-AN and R-AN are useless when the power ratio constraint in \eqref{eq_cor1_and_2_alpha} is violated.
\end{enumerate}

%---------------End Remark--------------------
%\begin{figure*}[!bpt]
%	\normalsize
%	\begin{equation}
%	\underline{R}_s^{\text{UQ}} =
%	\begin{cases} 
%	\left[	\log \left(1+ \frac{\operatorname{tr^{-1}}(\vect{\Xi} )  \vartheta_{k}^4 \theta  \beta _k  p_d  N}{\theta  p_d\beta _k + P_{k, \textnormal{UQ}}^{\textnormal{AN}} +1}\right) - \log \left(1 + \frac{\vartheta_{k} ^2 \theta  p_d \beta _e \left(\vartheta _{k}^2 \kappa_R  N+1\right)}{\text{tr}(\vect{\Xi}  ) ( P_{e, \textnormal{UQ}}^{\textnormal{AN}}+1) }\right) \right]^+ & \text{MRT-BF} \\ 
%	\left[\log \left(1+\frac{\theta  p_d \beta _k (N-K) \operatorname{tr^{-1}}(\vect{\Xi^{-1}} ) }{\theta  p_d \beta _k (1-\vartheta_{k}^2) +P_{k, \textnormal{UQ}}^{\textnormal{AN}} +1}\right) -\log \left(1+\frac{\theta  p_d \beta _e \left({\vartheta _{k}^{-2}}+\kappa_R (N-K-1)\right)}{\text{tr}\left({\vect{\Xi^{-1}} }\right) (P_{e, \textnormal{UQ}}^{\textnormal{AN}}  +1) }\right)  \right]^+&\text{ZF-BF}  
%	\end{cases}
%	\label{eq:Rs_UQ}
%	\end{equation}
%	\hrulefill
%	\vspace*{0pt}
%\end{figure*}

%==========================================================
\subsection{Massive MIMO with power scaling}
\label{sec4.2}
%==========================================================

Here, we assume that as $N \to \infty$, the transmit power at the BS can be scaled down by a factor of  $1/\sqrt{N}$ or $1/N$.  In the sequel, we use PS1 and PS2 to denote the situations where the transmit power is proportional to $1/\sqrt{N}$ and $1/N$, respectively. Hence,
\begin{equation}
p_d=
\begin{cases}
{\rho}\big /{\sqrt{N}}&\textnormal{PS1}\\
{\rho} \big/{{N}}  &\textnormal{PS2}
\end{cases}
\end{equation}
where $\rho$ is a fixed value (predetermined at the BS). 

We state our results in the following two corollaries. 
\begin{cor}\label{cor2}
	Consider the BS's transmit power is scaled down by a factor of $1/\sqrt{N}$.  If the BS employs MRT-BF or ZF-BF, then the maximum secrecy rate converges to
	\begin{equation}\label{eq_cor3_1}
	\underline{R}_s^{ \textnormal {PS1}} \to \left [ \log \left ( \frac{  \beta _k }{  \kappa_R \beta_e }\right ) \right]^{+}
	\end{equation}
	irrespective of the artificial noise scheme.
\end{cor}
\begin{proof}
The result is established by taking the limits of $\underline{R}_s^{\text{ZF}}(\rho/\sqrt{N}, \theta, N)$ \eqref{eq_thm12_sec_rate2}  and $\underline{R}_s^{\text{MRT}}(\rho/\sqrt{N}, \theta, N)$ \eqref{eq_thm12_sec_rate1} as $N \to \infty$ and then maximizing the resulting expressions with respect to $\theta$, following the same reasoning as the proof for Corollary \ref{cor1}. 
\end{proof}
%-------------------------------------------
\begin{cor} \label{cor3}
		Consider the BS's transmit power is scaled down by a factor of $1/{N}$. If the BS employs MRT-BF, then the maximum secrecy rate converges to
	\begin{equation}\label{eq_cor4_1}
	\underline{R}_s^{\textnormal{PS2, MRT}} \to \left [ \log \left ( \frac{\pi  \operatorname{tr} (\vect{\Sigma} )+2 \beta _k \sigma_{\hat{h}_{k}}^4 \rho      }{\pi  \operatorname{tr} (\vect{\Sigma} ) +2 \beta _e     \sigma_{\hat{h}_{k}}^4  \kappa_R \rho }\right ) \right]^{+}
	\end{equation}
	and when the BS employs ZF-BF, the maximum secrecy rate converges to
	\begin{equation} \label{eq_cor4_2}
	\underline{R}_s^{\textnormal{PS2, ZF}} \to \left [ \log \left ( \frac{ \pi \operatorname{tr} (\vect{\Sigma}^{-1})  + 2 \beta _k \rho   }{\pi  \operatorname{tr} (\vect{\Sigma}^{-1}) +2\beta_e \kappa_R \rho      \  }  \right ) \right]^{+}
	\end{equation}
	irrespective of the artificial noise scheme.
\end{cor}
\begin{proof}
The result is established by taking the limits of $\underline{R}_s^{\text{ZF}}(\rho/N, \theta, N)$ \eqref{eq_thm12_sec_rate2}  and $\underline{R}_s^{\text{MRT}}(\rho/N, \theta, N)$ \eqref{eq_thm12_sec_rate1} as $N \to \infty$ and then maximizing the resulting expressions with respect to $\theta$, following the same reasoning as the proof for Corollary \ref{cor1}. 
\end{proof}

By inspection of Corollaries \ref{cor2} \& \ref{cor3} we can observe that a positive secrecy rate is possible if the transmit power ratio satisfies
\begin{equation} \label{eq_TPR2}
\kappa_T = \frac{p_e}{p_k}< \left(\frac{\beta_k}{\beta_e} \right)^2.
\end{equation}
%This means that the transmit power ratio $\kappa_T $ plays a %central role in impacting the secrecy rate. 
Since \eqref{eq_cor4_1} and \eqref{eq_cor4_2}  are both positive under the same condition \eqref{eq_TPR2}, thus it is easy to show that 
\begin{equation}\label{eq_cor1_2_difference2}
\Delta^{\textnormal{PS2}}= \underline{R}_{s}^{\textnormal{PS2, ZF}}-\underline{R}_{s}^{\textnormal{PS2, MRT}} >0
\end{equation}
asymptotically.

%----------------------Remark---------------
We summarize our conclusions from Corollaries \ref{cor2} \& \ref{cor3} as follows:
%[wide, labelwidth=!, labelindent=0pt]
\begin{enumerate}
	\item  When power scaling at the BS is considered, the asymptotic performance is independent of artificial noise, contrary to the no-power-scaling regime.
	\item Under PS1, MRT-BF and ZF-BF are equivalent while under PS2 regime, ZF-BF  outperforms MRT-BF, asymptotically.
	\item With power scaling at the BS, the asymptotic secrecy rate drops to zero  when \eqref{eq_TPR2} is violated.
\end{enumerate}

Finally, we close this section with the following.
\begin{remark}
Corollaries \ref{cor1}--\ref{cor3} can be used to deduce the asymptotic secrecy rate for the passive eavesdropping case and the unquantized system, considering the variable replacements as discussed in Sec. \ref{sec3.6} and Sec. \ref{sec3.7}.
\end{remark}

%==========================================================
\section{Numerical Results and discussion} 
\label{sec5}
%==========================================================

In this section, we present some numerical results to verify the analytical results in this work. We consider a single-cell Massive MIMO system with $K$ single-antenna users and a single-antenna active eavesdropper. Without loss of generality, we assume $\beta_1=\beta_2=\cdots =\beta_K=\beta_{e}=1$ and all legitimate users transmit at the same power, i.e., $p_1 = p_2=\cdots=p_k = p_u$. Unless otherwise stated, analytical results refer to the achievable secrecy rate using \eqref{eq_thm12_sec_rate2} and \eqref{eq_thm12_sec_rate1} and Corollaries \ref{cor1}-\ref{cor3} whereas simulation results refer to simulated achievable secrecy rate evaluated by Monte Carlo simulation with quantization-noise correlation and exact ergodic information rate leakage \eqref{eq10.3} are accounted. 

\begin{figure}[!ht]
	\centering
	\includegraphics[scale=0.6]{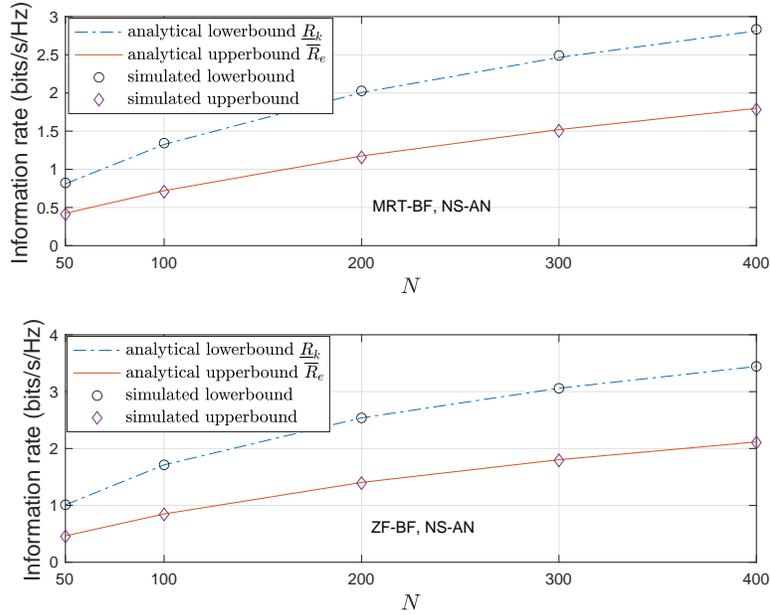}
	\caption{Theoretical and simulated lower bound on achievable user rate and upper bound on eavesdropper rate under MRT-BF (topmost plot) and ZF-BF (bottommost plot) with NS-AN. We use $\theta=0.5$, $K=\tau= 10$, $ p_u = p_d = 10$dB and $p_{e}=5$dB. The theoretical results are obtained by \eqref{thm1_MRT},\eqref{thm1_ZF}, \eqref{thm2_eq1}, and \eqref{thm2_eq2}.}
\label{fig: 0}
\end{figure}

%\Figure[htp!](topskip=0pt, botskip=0pt, midskip=0pt)[scale=0.43]{./figures/LB_UB_ver.eps}{Theoretical and simulated lower bound on achievable user rate and upper bound on eavesdropper rate under MRT-BF (topmost plot) and ZF-BF (bottommost plot) with NS-AN. We use $\theta=0.5$, $K=\tau= 10$, $ p_u = p_d = 10$dB and $p_{e}=5$dB. The theoretical results are obtained by \eqref{thm1_MRT},\eqref{thm1_ZF}, \eqref{thm2_eq1}, and \eqref{thm2_eq2}. \label{fig: 0}}

In Fig. \ref{fig: 0} we show the simulated and theoretical lower and upper bounds on information rates of the legitimate user and eavesdropper, respectively.  We show the results assuming the BS employs nullspace noise.  The topmost and bottommost plots compare the theoretical (using \eqref{thm1_MRT},\eqref{thm1_ZF}, \eqref{thm2_eq1} and \eqref{thm2_eq2}) bounds with simulated ones under MRT-BF and ZF-BF, respectively. As can be seen from Fig. \ref{fig: 0}  that there is a good match between the analytical and the simulated results. 

%\begin{figure}[ht!]
%	\centering
%	\begin{subfigure}[]{\linewidth} 
%		\centering
%		\includegraphics[width=3.6in]{./figures/Fig1_MRT_1_1.eps}
%		\caption{MRT-BF with R-AN}
%		\label{fig:1-1}
%	\end{subfigure}
%	\begin{subfigure}[]{\linewidth}
%		\centering
%		\includegraphics[width=3.6in]{./figures/Fig2_MRT_1_2.eps}
%		\caption{MRT-BF with NS-AN}
%		\label{fig:1-2}
%	\end{subfigure}
%	\caption{Achievable secrecy rate under MRT-BF with (a) R-AN and (b) NS-AN,  $K=10, \tau=K, p_u = p_d = 10$dB and $p_{e}=7$dB.}
%	\label{fig:1} 
%\end{figure}

\begin{figure*}[!ht]
	\centering
	\vspace*{-.02in}
	\subfloat[MRT-BF with R-AN]{
		\hspace*{-.1in}
		\includegraphics[scale=0.6]{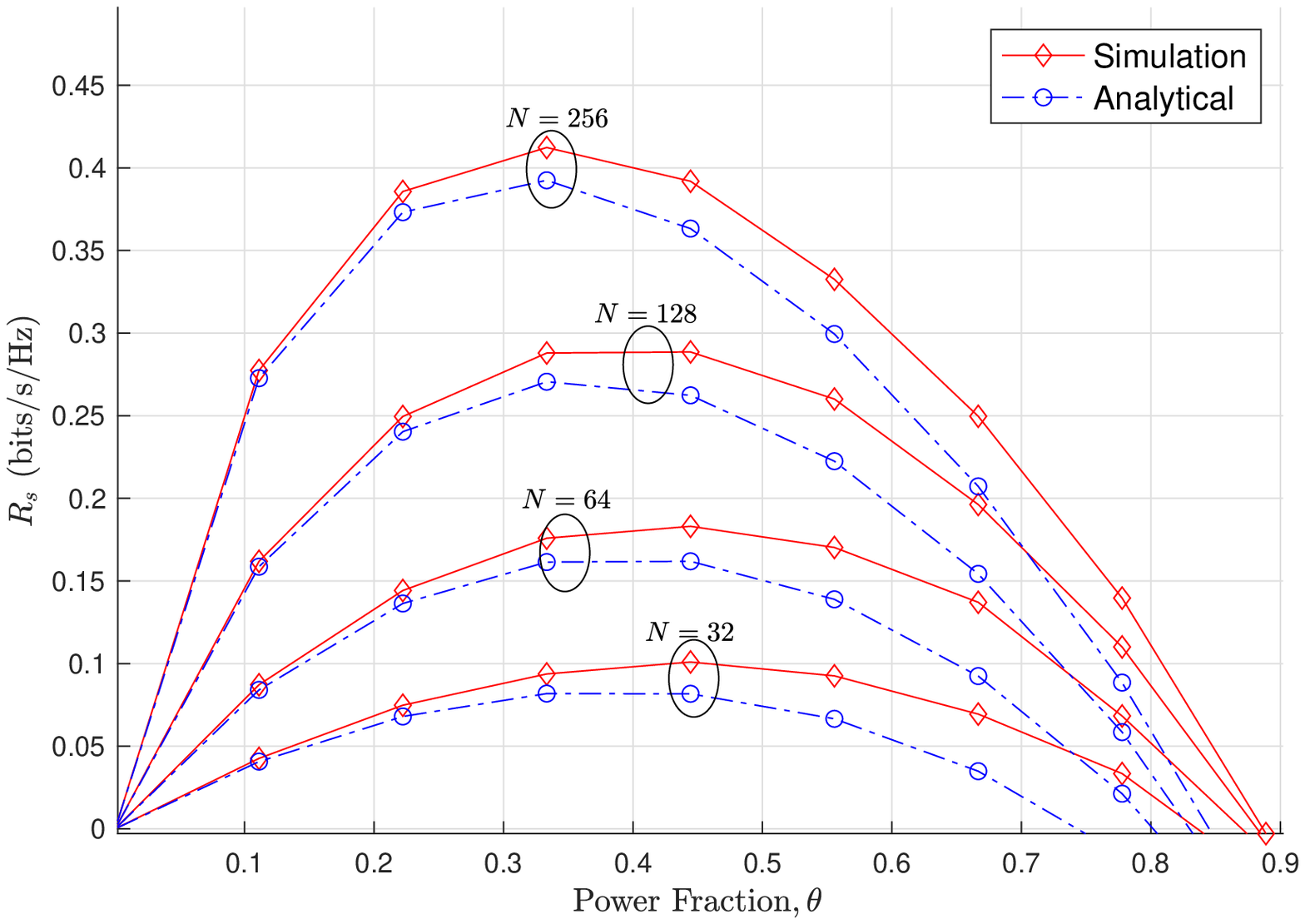}
		\label{fig:1.1}
	}
  %\vspace*{-.14in}
	\hfill
	\subfloat[MRT-BF with NS-AN]{
		\hspace*{-.1in}
		\includegraphics[scale=0.6]{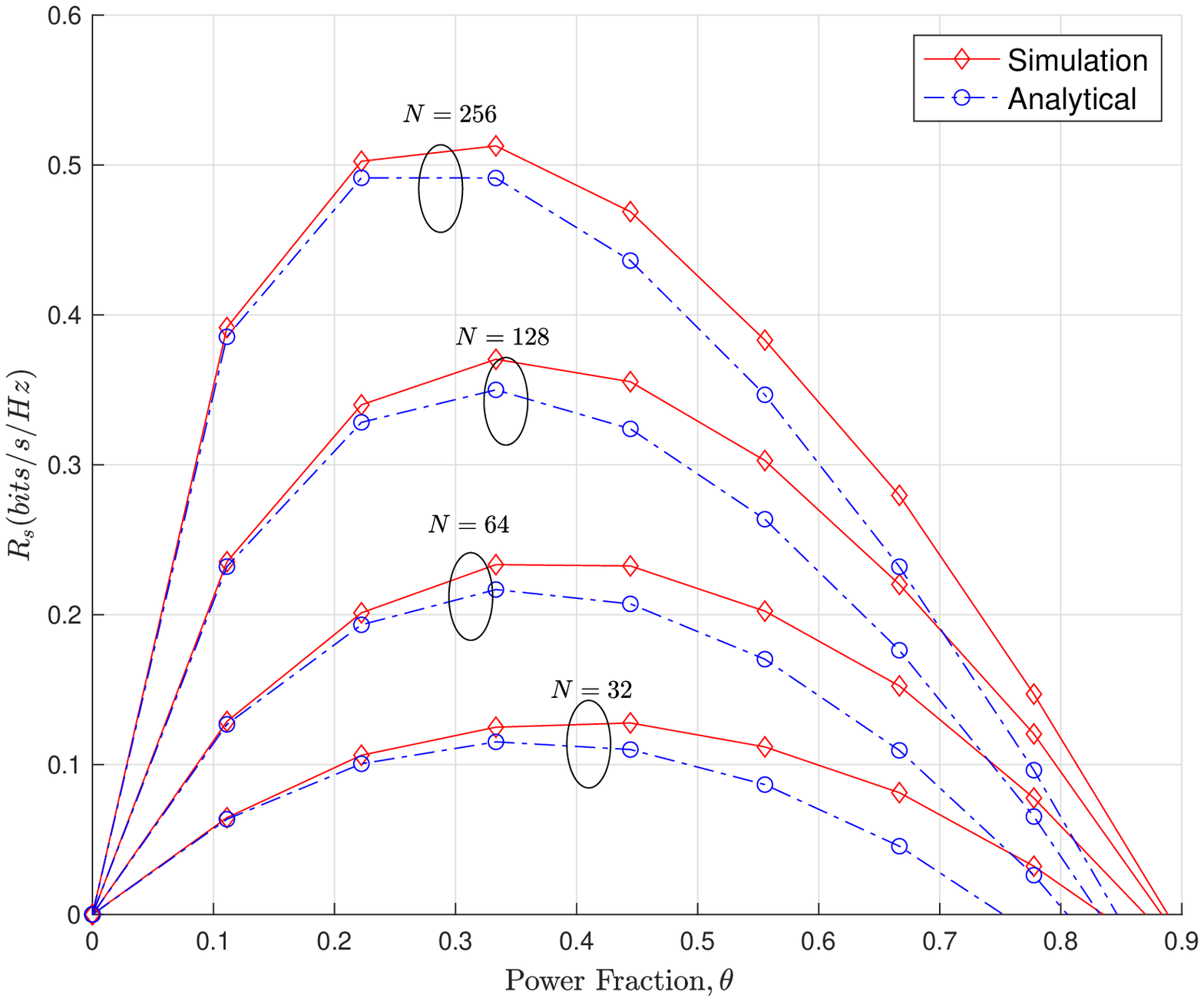}
		\label{fig:1.2}
	}
	\caption{Achievable secrecy rate of MRT-BF for different number of BS antennas, $K=\tau=10, p_u = p_d = 10$dB and $p_{e}=7$dB.}
	\label{fig:1}
\end{figure*}

The achievable secrecy rate corresponding to MRT-BF and ZF-BF is shown in Fig. \ref{fig:1} and Fig. \ref{fig:2}, respectively, for a different number of BS antennas ($N=32, 64, 128, 256$) and as varying $\theta$ (allocated power ratio of signal) between 0 and 1.  The eavesdropper's power is set to $p_e =p_u/2= 7$dB.  From Fig. \ref{fig:1}, we can observe that the NS-AN (Fig. \ref{fig:1.2})  always outperforms R-AN (Fig. \ref{fig:1.1}) due to a smaller leakage power of the artificial noise. For example, when $N=256$, the performance gap between NS-AN and R-AN is about 0.1 bits/s/Hz.  Likewise, it is evident from Fig. \ref{fig:2} that the use of  NS-AN (Fig. \ref{fig:2-2}) provides higher rates compared with R-AN (Fig. \ref{fig:2-1}) under ZF-BF. Further, it is clear that ZF-BF with NS-AN achieves the highest secrecy rate while MRT-BF with R-AN provides the lowest secrecy rate, where the gap between them is about 0.3 bits/s/Hz when $N=256$. We can observe that the analytical results serve as a good lower bound on the secrecy rate compared with the simulated results. It is worth noting that, in the neighborhood of the optimal value of $\theta$  where the secrecy rate is peaked, the gap between the analytical and simulated results is very small.

%\begin{figure}[ht!]
%	\centering
%	\begin{subfigure}[ht!]{\linewidth} 
%		\centering
%		\includegraphics[width=3.6in]{./figures/Fig1_ZF_2_1.eps}
%		\caption{ZF-BF with R-AN}
%		\label{fig:2-1}
%	\end{subfigure}
%	\begin{subfigure}[ht!]{\linewidth}
%		\centering
%		\includegraphics[width=3.6in]{./figures/Fig2_ZF_2_2.eps}
%		\caption{ZF-BF with NS-AN}
%		\label{fig:2-2}
%	\end{subfigure}
%	\caption{Achievable secrecy rate under ZF-BF with (a) R-AN and (b) NS-AN,  $K=10, \tau=K, p_u = p_d = 10$dB and $p_{e}=7$dB.}
%	\label{fig:2} 
%\end{figure}

\begin{figure*}[!ht]
	\centering
	\vspace*{-.02in}
	\subfloat[ZF-BF with R-AN]{
		\hspace*{-.1in}
		\includegraphics[scale=0.6]{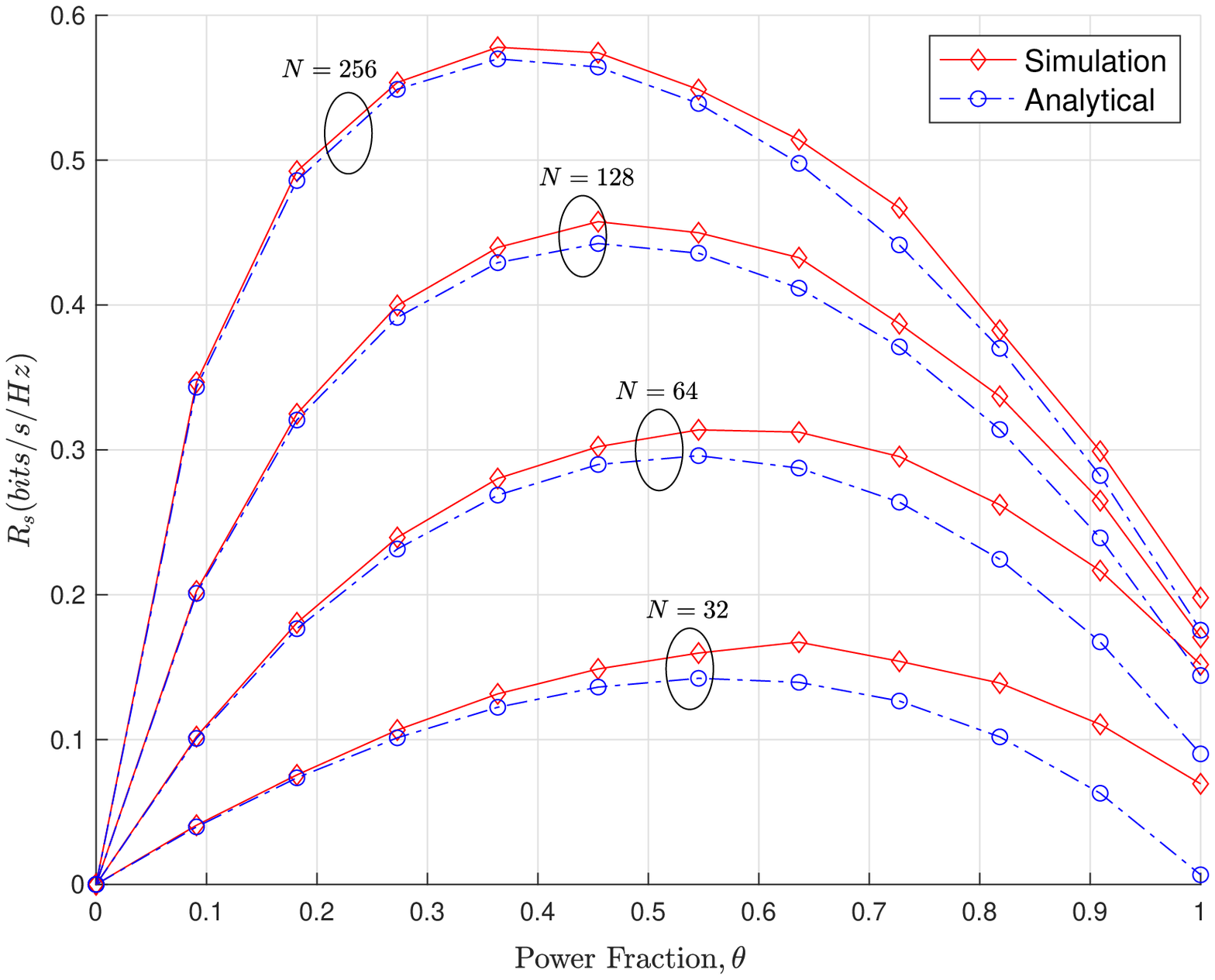}
		\label{fig:2-1}
	}
	%\vspace*{-.14in}
	\hfill
	\subfloat[ZF-BF with NS-AN]{
		\hspace*{-.1in}
		\includegraphics[scale=0.6]{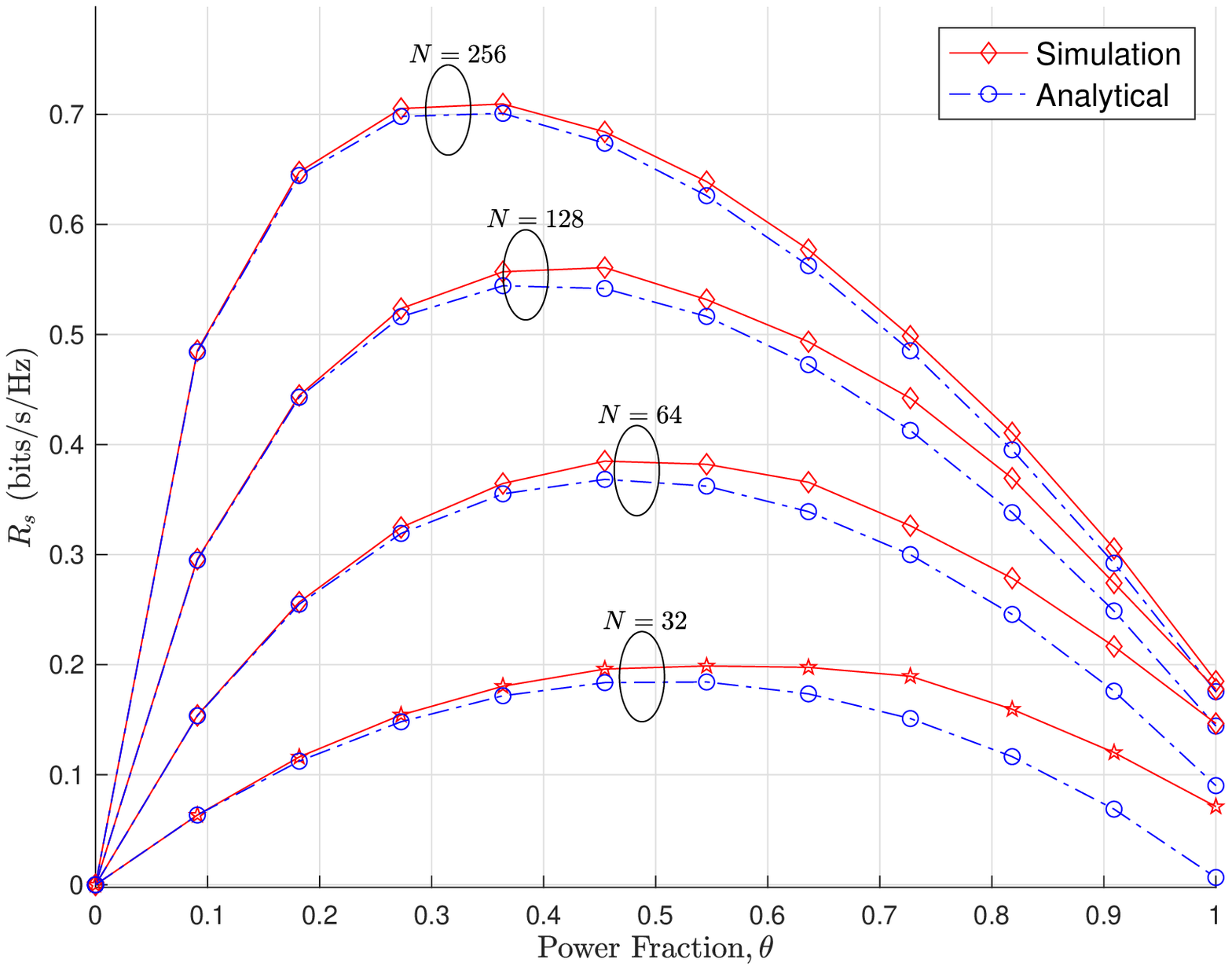}
		\label{fig:2-2}
	}
	\caption{Achievable secrecy rate of ZF-BF for different number of BS antennas, $K=\tau = 10, p_u = p_d = 10$dB and $p_{e}=7$dB.}
	\label{fig:2}
\end{figure*}

% As seen in Fig. \ref{fig:2}, a relatively smaller gap (compared with MRT-BF) between the simulated and analytical results. This is partly because the user rate improves under ZF-BF and hence this improvement will render the gap, resulting from our approximation error and the use of Jensen's inequality, smaller, i.e., see Appendix \ref{proof_thm2}. 
 
 Moreover, we observe that in all simulated cases in Figs. \ref{fig:1} and \ref{fig:2}, the secrecy rate increases as the number of BS antennas $N$ increases, while the power fraction allocated to signal is monotonically decreasing. As $N$ increases, both the intercepted user's rate and information leakage increase, thus in order to maintain a positive secrecy rate, more power should be allocated to artificial noise to degrade the eavesdropper channel (see the proof of Corollary \ref{cor1}).

\subsection{Impact of number of users}
\begin{figure}[!ht]
	\centering
	\includegraphics[scale=0.6]{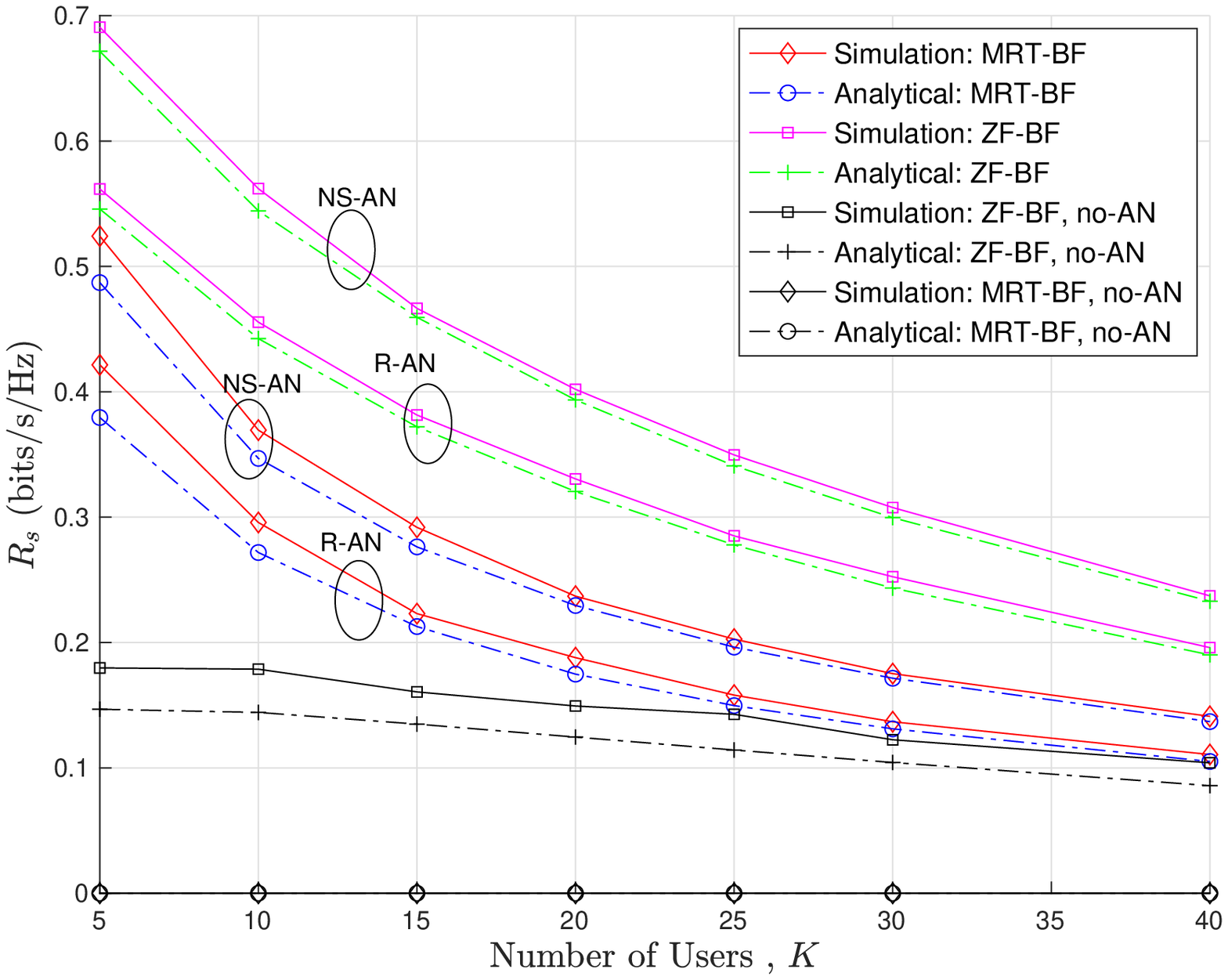}
	\caption{The impact of number of users on secrecy rate under ZF-BF and MRT-BF, $N=128$, $\tau=K$, $p_u = p_d = 10$dB and $p_{e} = 5$dB.}
\label{fig: 3}
\end{figure}

%\Figure[!ht](topskip=0pt, botskip=0pt, midskip=0pt)[scale=0.43]{./figures/effect_K.eps}{The impact of number of users on secrecy rate under ZF-BF and MRT-BF, $N=128$, $\tau=K$, $p_u = p_d = 10$dB and $p_{e} = 5$dB. \label{fig: 3}}

Fig. \ref{fig: 3} depicts the impact of increasing the number of users on the secrecy rate. As seen, the secrecy rate decreases steadily as the number of users increases. This, in particular, follows from the increases of inter-user interference (in case of MRT-BF) and the reduction in the array gain (in case of ZF-BF), thus reducing the rate of the intercepted user. As observed previously, ZF-BF with NS-AN provides a higher secrecy rate, albeit at the price of a high computational burden when compared with MRT-BF combined with R-AN. 

\subsection{Impact of transmit power ratio}
\begin{figure}[ht!]
	\centering
	\includegraphics[scale=0.6]{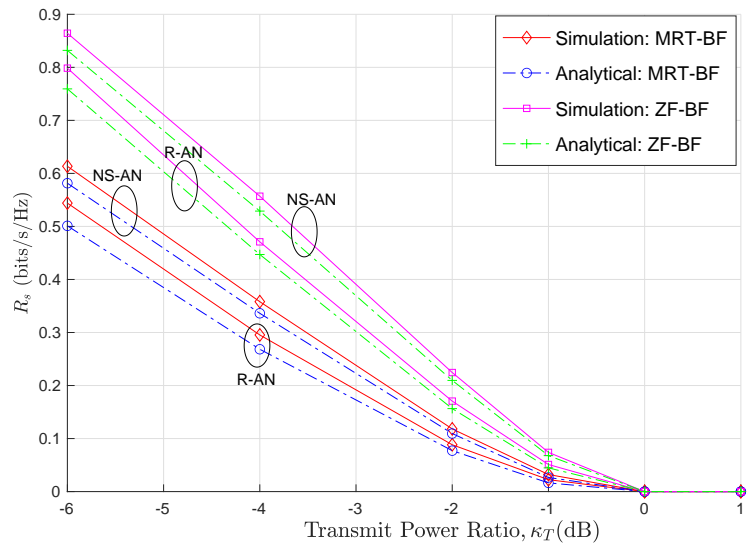}
	\caption{The impact of  transmit power ratio $\kappa_T = p_e/p_u$, during pilot attack, on secrecy rate, $N=64$, $K=\tau = 10$ and  $p_u = p_d = 10$dB.}
	\label{fig:4}
\end{figure}

%\Figure[!ht](topskip=0pt, botskip=0pt, midskip=0pt)[scale=0.44]{./figures/effect_kappa.eps}{The impact of  transmit power ratio $\kappa_T = p_e/p_u$, during pilot attack, on secrecy rate, $N=64$, $K=\tau = 10$ and  $p_u = p_d = 10$dB. 	\label{fig:4}}

The effect of transmit power ratio $\kappa_T = p_e/p_u$ during the pilot attack is illustrated in Fig. \ref{fig:4}. In all beamforming and artificial noise schemes, we observe that the secrecy rate is steadily reduced as $\kappa_T$ increases. In general, ZF-BF with NS-AN outperforms other schemes as observed previously. However, the secrecy rate drops to zero for all schemes when $\kappa_T$ approaches 1 (0dB). This is in line with the asymptotic condition derived in \eqref{eq_cor1_and_2_alpha}. From \eqref{eq_cor1_and_2_alpha}, $\kappa_T < 1$ due to $\beta_k = \beta_e = 1$ in our simulation. Thus in the absence of an advanced  secrecy protocol, active eavesdropping can be deleterious to the secrecy rate.

\begin{figure}[!ht]
	\centering
	\includegraphics[scale=0.6]{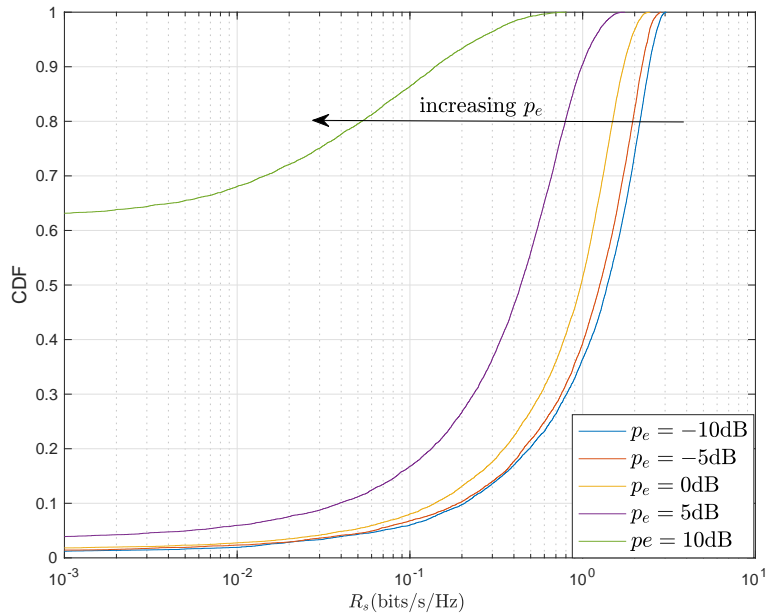}
	\caption{CDF of secrecy rate for (ZF-BF, NS-AN)-scheme. The BS is positioned in the center of a circle of radius $1$km while the eavesdropper is in a circle of radius $100$m around the intercepted user. All users' positions are random and uniformly distributed in the cell. $N=128$, $K=\tau = 10$  and $p_d= p_u= 10$dB.}
\label{fig:7}
\end{figure}

%\Figure[!ht](topskip=0pt, botskip=0pt, midskip=0pt)[scale=0.43]{./figures/cdf_Rs.eps}{CDF of secrecy rate for (ZF-BF, NS-AN)-scheme. The BS is positioned in the center of a circle of radius $1$km while the eavesdropper is in a circle of radius $100$m around the intercepted user. All users' positions are random and uniformly distributed in the cell. $N=128$, $K=\tau = 10$  and $p_d= p_u= 10$dB.	\label{fig:7}}

Fig. \ref{fig:7} depicts the cumulative distribution function (CDF) of the secrecy rate when the BS employs ZF-BF and NS-AN, where this scheme is chosen due to its high performance as we have shown before. We assume that the BS is positioned in the center of a circle of radius 1km while the active eavesdropper in a circle of radius 100m around the intercepted user, i.e., this captures the situation when the eavesdropper is very close to the intercepted user. The positions of users are assumed random and uniformly distributed inside the circular cell.  As seen in Fig. \ref{fig:7} the average secrecy rate decreases with increasing the power of eavesdropper. When the eavesdropper transmits at the same power level as the legitimate user, the average secrecy drops to zero. This again confirms our analysis and the transmit power-ratio threshold given in \eqref{eq_cor1_and_2_alpha} even in this non-asymptotic case.  

\subsection{Active Vs. Passive eavesdropping}
\begin{figure}[!ht]
	\centering
	\includegraphics[scale=0.6]{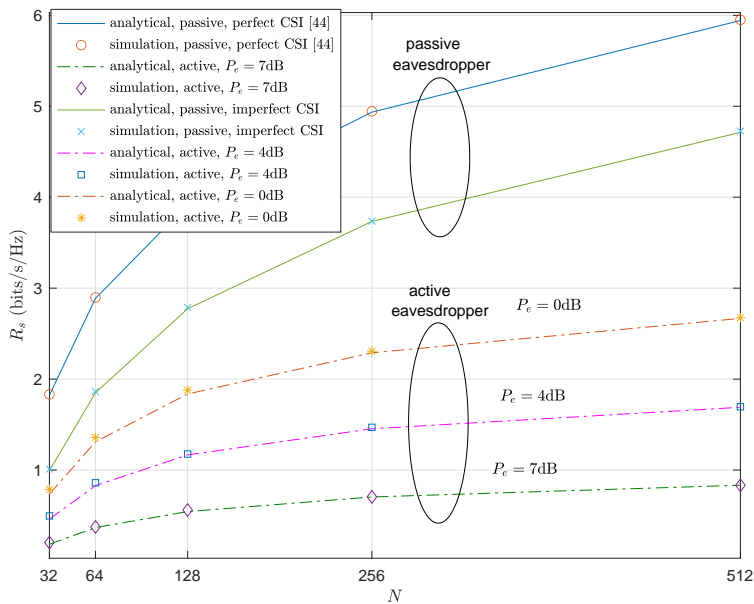}
	\caption{Passive vs. active eavesdropping performance comparison under ZF-BF and NS-AN. $K=\tau = 10$ and $p_d = p_u = 10$dB. For passive eavesdropping, we show the results for perfect and imperfect knowledge of CSI, i.e., P-ICSI and P-PCSI. The analytical results for passive eavesdropping use \eqref{eq:P-ICSI} and \eqref{eq:P-PCSI}.}
	\label{fig:4_1}
\end{figure}

%\Figure[!ht](topskip=0pt, botskip=0pt, midskip=0pt)[scale=0.43]{./figures/passive_vs_active.eps}{Passive vs. active eavesdropping performance comparison under ZF-BF and NS-AN. $K=\tau = 10$ and $p_d = p_u = 10$dB. For passive eavesdropping, we show the results for perfect and imperfect knowledge of CSI, i.e., P-ICSI and P-PCSI. The analytcal analysis for passive eavesdropper use \eqref{eq:P-ICSI} and \eqref{eq:P-PCSI}.
%		\label{fig:4_1}}

We plot in Fig. \ref{fig:4_1} the theoretical and simulated secrecy rate versus the number of BS antennas in the presence of active and passive eavesdropping.  We show the results for the case of ZF-BF and NS-AN. We can see that when the eavesdropper is passive (i.e., $P_e=0$ or $\kappa_R=0$), the secrecy rate increases monotonically with the number of BS antennas under both perfect and imperfect knowledge of CSI at the BS. This is in line with our analytical expressions in \eqref{eq:P-ICSI}  and \eqref{eq:P-PCSI}. We remark again that in unquantized Massive MIMO systems, the monotonic increase of secrecy rate in the presence of a passive single-antenna eavesdropper has been reported in the literature (i.e., see \cite{Kapetanovic2015}) and it also holds true for quantized systems. 

However, when an active eavesdropper exists, the secrecy rate becomes drastically small and grows at a much slower pace as the number of BS antennas $N$ increases, especially when the pilot attack is strong. Guided by \eqref{eq_cor2} of Corollary \ref{cor1}, the slow growth of secrecy rate with increasing $N$ indicates that the secrecy rate will finally saturate as $N \to \infty$.  The performance gap between the passive and active eavesdropping scenarios is significant, even when the transmit power of the eavesdropper is at the noise level. 

 Finally, the results of passive eavesdropping demonstrate just how the assumption of perfect CSI may overestimate the achievable secrecy rate.

\subsection{Quantized Vs. Unquantized systems}
\begin{figure}[!ht]
	\centering
	\includegraphics[scale=0.6]{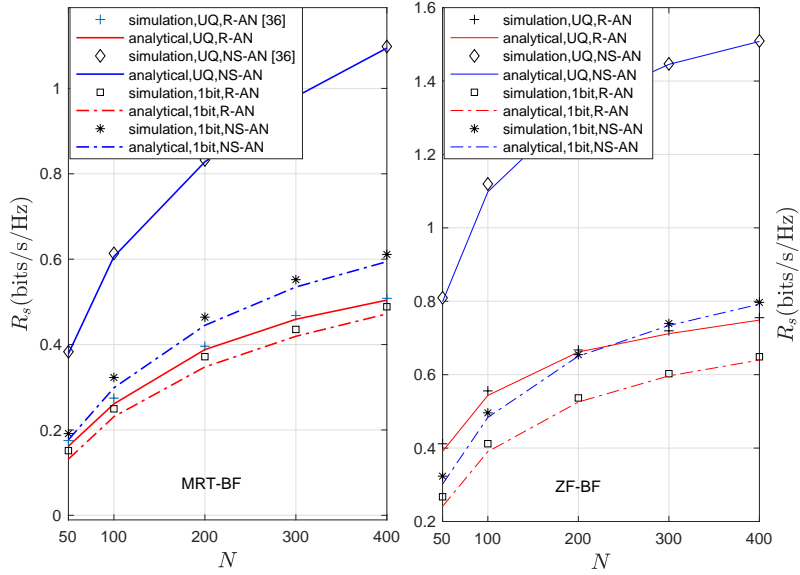}
	\caption{Performance comparison between the one-bit quantized and unquantized systems under MRT-BF (leftmost plot) and ZF-BF(rightmost plot). $K= \tau = 10$ and $p_d = p_u = 10$dB, $P_e =7$dB. The analytical results for the unquantized system use \eqref{eq:Rs_UQ}. Dash-dotted and solid lines refer to the one-bit quantized and unquantized systems, respectively.}
		\label{fig:4_2}
\end{figure}

%\Figure[!ht](topskip=0pt, botskip=0pt, midskip=0pt)[scale=0.43]{./figures/QUQfig.eps}{Performance comparison between the one-bit quantized and unquantized systems under MRT-BF (leftmost plot) and ZF-BF(rightmost plot). $K= \tau = 10$ and $p_d = p_u = 10$dB, $P_e =7$dB. The analytical results for the unquantized system use \eqref{eq:Rs_UQ}. Dash-dotted and solid lines refer to the one-bit quantized and unquantized systems, respectively.
%	\label{fig:4_2}}

In Fig.\ref{fig:4_2}, we show the theoretical and simulated secrecy rate for both the one-bit quantized and unquantized (i.e., infinite-resolution ADCs/DACs at the BS) systems. The impact of quantization noise on the achievable secrecy rate is captured by the performance gap between the two systems. The larger the gap, the larger the effect of quantization noise and vice versa. We show the results for MRT-BF (leftmost plot) and ZF-BF (rightmost plot). It is clear that the secrecy rate achieved by the unquantized system is larger for all simulated $N$, especially when the ZF-BF is employed.  For example, compared with the unquantized system, the quantized system requires roughly eight times the number of antennas to achieve 0.8 bits/s/Hz when using  ZF-BF and NS-AN, while it requires four times the number of antennas to achieve 0.6 bits/s/Hz when using MRT-BF and NS. However, when R-AN is used,  the performance gap becomes relatively small, when compared with the gap resulting from using NS-AN. Thus the use of random artificial noise renders the performance loss due to quantization noise smaller.

\subsection{Asymptotic behaviour of secrecy rate}

In this subsection, we demonstrate the behavior of secrecy rate as $N \to \infty$. For all simulated results, it is entirely understood that allowing high order of magnitudes of $N$ is used only to show the correct asymptotic behavior. The simulated results are only shown for a conceivable number of BS antennas, i.e., $N=32, 64, 128, 512$. The asymptotes for the quantized system are given in Corollaries \ref{cor1}-\ref{cor3}. And the asymptotes for the unquantized system is derived from Corollaries \ref{cor1}-\ref{cor3} while considering the parameters' replacement in \eqref{eq:Variable_change2}.

Fig. \ref{fig:5} illustrates the asymptotic behavior of the secrecy rate as $N \to \infty$. As seen, when no power scaling is used at the BS (topmost plot), both MRT-BF and ZF-BF are asymptotically equivalent. As $N$ gets larger and larger, almost all power is allocated to artificial noise asymptotically ($\theta \to 0$), thus the artificial noise being used dominates (determines) the performance asymptotically. We can observe that under the no-PS case, NS-AN outperforms R-AN. 

When the BS's power is scaled down by $N$ (bottommost plot), almost all power should be allocated to data ($\theta \to 1$) to maintain a positive secrecy rate as $N \to \infty$, rendering both R-AN and NS-AN equivalent asymptotically, and hence the beamforming scheme being used determines the performance. It is evident that the ZF-BF outperforms MRT-BF. When the power scales down with $\sqrt{N}$ (middlemost plot),  any combinations of beamforming and artificial noise schemes are asymptotically equivalent. The reader will observe the very large number of BS antennas for the no-PS and PS1 cases to converge to the corresponding asymptotic values, compared with the PS2 case which converges at a much faster pace.

\begin{figure}[!htp]
	\centering
	\includegraphics[scale=0.7]{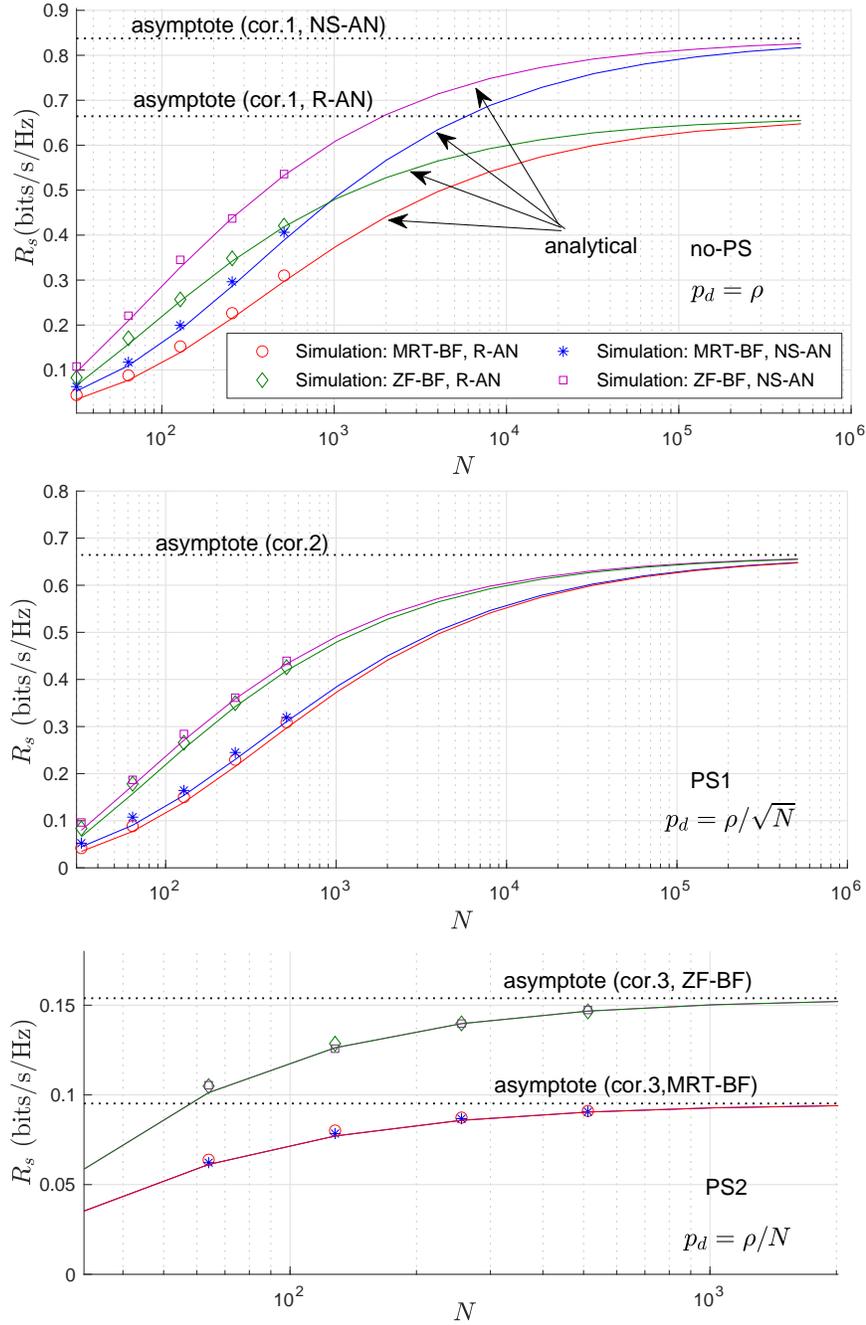}
	\caption{ The asymptotic behaviour of secrecy rate under no-PS (topmost plot), PS1 (middlemost plot) and PS2 (bottommost plot) power scaling regimes. We use  $K= \tau = 10, p_u =10$dB, $\rho = 10$ dB(fixed power at BS) and $\kappa_T = -2 $dB ($p_e = 8$dB).  The three scenarios, no-PS, PS1 and PS2 correspond, respectively, to $p_d = \rho$,  $p_d=\rho/\sqrt{N}$ and  $p_d=\rho/N$. Markers, solid lines and dotted lines represent simulated, analytical and asymptotic results, respectively.}
		\label{fig:5} 
\end{figure}

%\Figure[!ht](topskip=0pt, botskip=0pt, midskip=0pt)[scale=0.47]{./figures/asymptotes.eps}
%{The asymptotic behaviour of secrecy rate under no-PS (topmost plot), PS1 (middlemost plot) and PS2 (bottommost plot) power scaling regimes. We use  $K= \tau = 10, p_u =10$dB, $\rho = 10$ dB(fixed power at BS) and $\kappa_T = -2 $dB ($p_e = 8$dB).  The three scenarios, no-PS, PS1 and PS2 correspond, respectively, to $p_d = \rho$,  $p_d=\rho/\sqrt{N}$ and  $p_d=\rho/N$. Markers, solid lines and dotted lines represent simulated, analytical and asymptotic results, respectively. 
%	\label{fig:5} 
%}

\begin{figure}[!htp]
	\centering
	\includegraphics[scale=0.6]{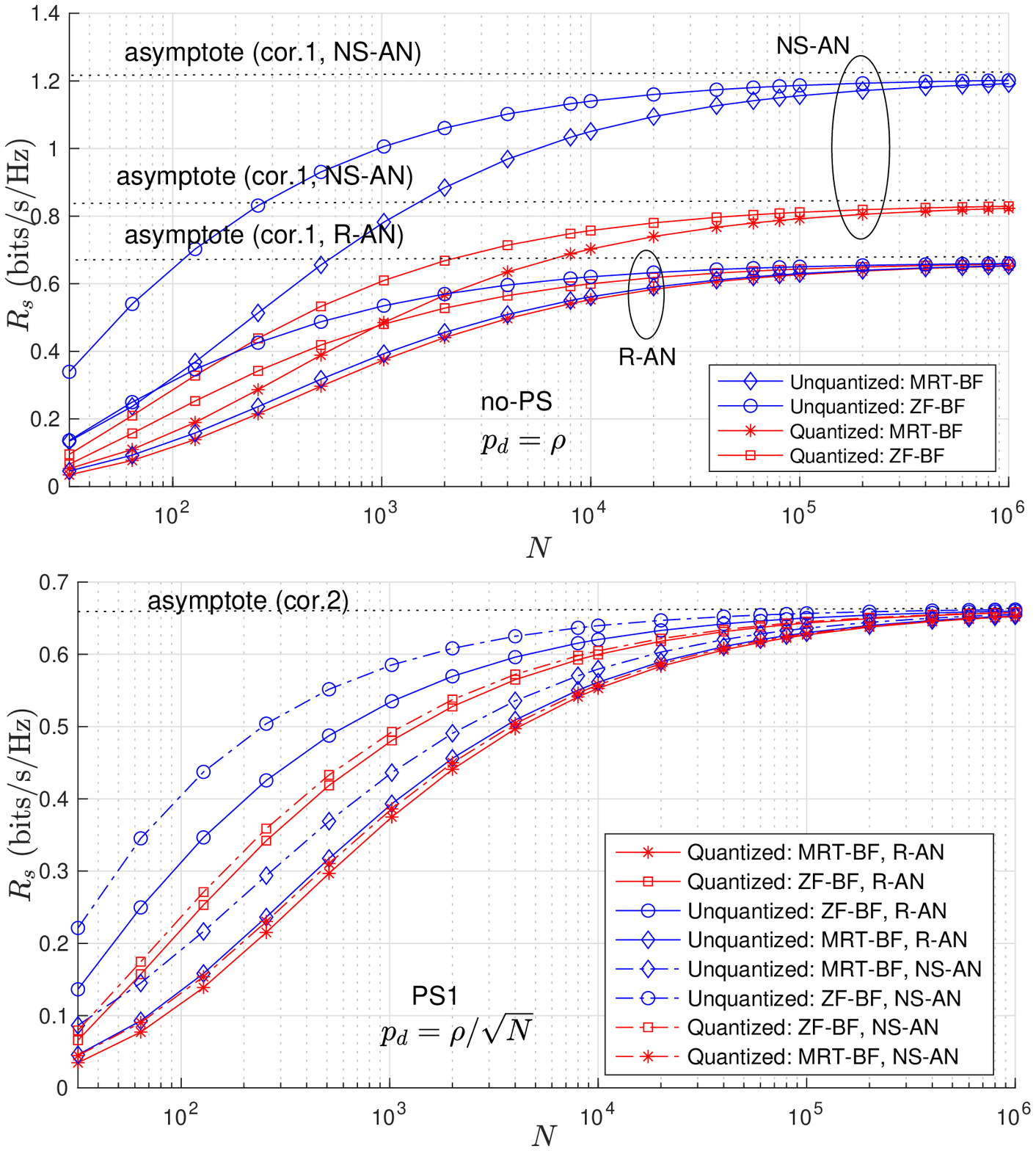}
	\caption{The asymptotic gap of secrecy rate between quantized and unquantized systems under no-PS (topmost plot) and PS1 (bottommost  plot) power scaling regimes. We use  $K= \tau = 10, p_u =10$dB, $\rho = 10$ dB(fixed power at BS) and $\kappa_T = -2 $dB ($p_e = 8$dB).  The two scenarios, no-PS and PS1 correspond, respectively, to $p_d = \rho$ and $p_d=\rho/\sqrt{N}$.}
	\label{fig:6}
\end{figure}
%
%\Figure[!ht](topskip=0pt, botskip=0pt, midskip=0pt)[scale=0.47]{./figures/Q_UQ_final.eps}
%{The asymptotic gap of secrecy rate between quantized and unquantized systems under no-PS (topmost plot) and PS1 (bottommost  plot) power scaling regimes. We use  $K= \tau = 10, p_u =10$dB, $\rho = 10$ dB(fixed power at BS) and $\kappa_T = -2 $dB ($p_e = 8$dB).  The two scenarios, no-PS and PS1 correspond, respectively, to $p_d = \rho$ and $p_d=\rho/\sqrt{N}$.
%		\label{fig:6}}

Fig. \ref{fig:6} shows the asymptotic performance gap between the quantized system and its unquantized (i.e., infinite-resolution ADCs/DACs) counterpart under no-PS and PS1 power scaling regimes. For the no-PS case, we observe a comparably larger gap when NS-AN is used whereas it is smaller when R-AN is used, especially under MRT-BF. Thus when the combination of MRT-BF and R-AN is considered, there is not much loss in secrecy rate due to quantization noise.  We also observe from Fig. \ref{fig:6} (topmost) that both quantized and unquantized systems are asymptotically equivalent under R-AN, in contrast to NS-AN. This implies that the leakage power of R-AN dominates the power of quantization noise, whereas the power of quantization noise dominates the leakage power of NS-AN in the asymptotic limit. 
For the PS1 regime in Fig. \ref{fig:6} (bottommost), the gap diminishes asymptotically under all schemes and hence quantization noise is irrelevant. 

To show that analytically, we present only the case of ZF-BF with NS-AN. The asymptotic rate corresponding to the unquantized system can be derived from \eqref{eq_cor2} and \eqref{eq_cor3_1} in Corollaries 1 and 2, respectively,  with parameters' change in \eqref{eq:Variable_change2}. Thus using \eqref{eq:Variable_change2} in \eqref{eq_cor2} and \eqref{eq_cor3_1}  we get
\begin{equation}
\label{eq_cor222}
\underline{R}_{s, \textnormal{NS-AN}}^{\textnormal{no-PS, UQ}} \rightarrow \left[ \log\left(\frac{\beta _k \left(p_d \beta _e \left(1-\kappa_R  \vartheta_k^2+1\right)\right)}{\kappa_R  \beta _e \left(p_d \left(1- \vartheta_k^2\right) \beta _k+1\right)} \right)\right]^{+}
\end{equation}

\begin{equation}
\label{eq_cor3_111}
\underline{R}_s^{ \textnormal {PS1, UQ}} \to \left [ \log \left ( \frac{  \beta _k }{  \kappa_R \beta_e }\right ) \right]^{+}
\end{equation}
as $N\to \infty$. It is clear that  $\underline{R}_{s, \textnormal{NS-AN}}^{\textnormal{no-PS, UQ}} > \underline{R}_{s, \textnormal{NS-AN}}^{\textnormal{no-PS}}$  and $\underline{R}_s^{ \textnormal {PS1, UQ}}=\underline{R}_s^{ \textnormal {PS1}}$. Therefore, under the no-PS regime, both secrecy rates of quantized and unquantized systems converge to distinct limits while under PS1 regime, both systems converge to the same limits, as $N\to \infty$.

Finally, for the case of PS2 which is not shown here, one can verify that the secrecy rate for the unquantized system converges to different asymptotic limits for ZF-BF and MRT-BF where the artificial noise scheme is asymptotically irrelevant.

%==========================================================
\section{Conclusion}
\label{sec6}
%==========================================================

This paper has investigated the secrecy in the downlink of Massive multiple-input multiple-output (MIMO) system under the presence of a single-antenna active eavesdropper and when the signal at the base station undergoes one-bit quantization. We investigated the efficacy of two artificial noise techniques; nullspace artificial noise (NS-AN) and random artificial noise (R-AN). Thus, we have derived the achievable secrecy rate when the BS uses the maximum-ratio transmission beamforming  (MRT-BF) and zero-forcing beamforming (ZF-BF). Although the very coarse quantization and pilot attack, secure communication is possible, where the best performance is achieved when ZF-BF is combined with NS-AN. In fact, we showed analytically that when the eavesdropper is sufficiently close to the intercepted user, the average secrecy rate drops to zero as the transmit power ratio between the eavesdropper and intercepted user approaches 1. The practical scenario examined in the paper has further corroborated our analysis.

It was shown that when the number of BS antennas $N$ grows large, the performance is independent of the beamforming technique and hence the NS-AN should be exploited to maximize the performance. This observation has an implication for research into other possible schemes of artificial noise to degrade the channel of the eavesdropper. Further, it was shown that the total power at the BS can be reduced proportional to $1/N$ or $1/\sqrt{N}$ while a positive secrecy rate is maintained, given the ratio between the eavesdropper's power and intercepted use's power is less than $(\beta_k/\beta_e)^2$, where $\beta_k$ and $\beta_e$ denote the large-scale fading coefficients of legitimate user (intercepted) and eavesdropper, respectively. This observation suggests considering other approaches other than artificial noise to enhance secrecy. 

Due to the scope limitation of this work, a number of potential issues needs to be considered in the future, such as power control and optimal design of beamforming. We believe our findings add to the understanding of the impact of active eavesdropping in quantized Massive MIMO systems.

\begin{appendices} 
	
	%==========================================================
	\section{Proof of Lemma 1}
	 \label{Proof_lemma1}
	 %==========================================================
	 
	From \eqref{eq5}, the channel estimate $\hat{\vect{h}}_k$ may be written as
	\begin{equation} \label{eq_lem1}
	\hat{\vect{h}}_k = \lambda_k \left(\sqrt{\gamma^2 p_k^{\prime} \tau} \vect{h}_k + \sqrt{\gamma^2 p_{e}^{\prime} \tau}\vect{g} + \gamma  \tilde{\vect{z}} + \tilde{\vect{q}} \right )
	\end{equation}
	where $\tilde{\vect{z}} \sim \mathcal{CN} (\vect{0},\vect{I}_N)$ and $\tilde{\vect{q}}\sim \mathcal{CN}(\vect{0},\sigma_q^2 \vect{I}_N)$. Further, since the eavesdropper's channel $\vect{g}$ and channel estimate $\hat{\vect{h}}_k$ are correlated, we can express $\vect{g}$ as 
	\begin{equation} \label{g_ghat}
\vect{g} = \hat{\vect{g}} + \vect{\epsilon}
	\end{equation}
where $\hat{\vect{g}}$ the optimal MMSE solution for $\vect{g}$ based on the observation $\hat{\vect{h}}_k$ and $\vect{\epsilon}$ is uncorrelated estimation error with minimum variance. It follows that
	\begin{equation}\label{eq_lem1.1}
	\hat{\vect{g}}= E\left[\vect{g} \hat{ \vect{h}}_k^H\right] \left (E\left[\hat{ \vect{h}}_k \hat{ \vect{h}}_k^H\right]\right)^{-1}	\hat{\vect{h}}_k
	\end{equation}
	From \eqref{eq_lem1} we have
	\begin{subequations}
		\begin{equation} \label{eq_lem1.2}
		E[\vect{g} \hat{ \vect{h}}_k^H] =  \sqrt{\lambda_k^2 \gamma^2 p_{e}^{\prime} \tau} \vect{I}_N 
		\end{equation}
		\begin{equation}\label{eq_lem1.3}
		E[\hat{ \vect{h}}_k \hat{ \vect{h}}_k^H]=\lambda_k^2( \gamma^2 p_k^{\prime} \tau + \gamma^2 p_{e}^{\prime}   \tau + \gamma^2 + \sigma_q^2) \vect{I}_N 
		\end{equation}		
	\end{subequations}	
	Substituting \eqref{eq_lem1.2} and \eqref{eq_lem1.3} with the definition of $\lambda_k$ \eqref{eq4.2} in \eqref{eq_lem1.1} yields $\hat{\vect{g}}= \sqrt{\kappa_R} \hat{ \vect{h}}_k $,  i.e., the first term of \eqref{lemma1_eq1}. 
	
	From \eqref{g_ghat} we write $\vect{\epsilon} = \vect{g} - \hat{\vect{g}}$.  From the orthogonality principle, the covariance matrix of estimation error is given by $\vect{C}_{\vect{\epsilon}} = E[\vect{\epsilon} \vect{\epsilon}^H] =E[(\vect{g} - \hat{\vect{g}}) \vect{g}^H]= (1-{\kappa_R} \sigma_{\hat{h}_{k}}^2 ) \vect{I}_N $, which is given in \eqref{lemma1_eq2}.
	This completes the proof.

%==========================================================	
\section{Proof of Theorem 1} 
\label{proof_thm1}
%==========================================================

In this section, we derive the two lower bounds \eqref{thm1_MRT}, \eqref{thm1_ZF}  on the achievable data rate of legitimate user $k$ (intercepted) while assuming the BS employs MRT-BF and ZF-BF schemes, respectively.  

From Lemma \ref{R_k_general}, the achievable data rate for any beamforming scheme is again given by
	\begin{equation}\label{eq10.2_1}
	\underline{R}_k= \log\left(1+\frac{|a|^2}{\sigma_{n_{\text{eff}}}^2}\right)
	\end{equation}

where $a$ and $\sigma_{n_{\text{eff}}}^2$ are defined by
	\begin{align}
	\label{constant_a}
	a &= c_1 \sqrt{\beta_{k}} E[ \vect{h}_k^T  \vect{w}_k] \\ \nonumber
	\label{var_noise}
	\sigma_{n_{\text{eff}}}^2&= c_1^2 \beta_k \operatorname{Var} (\vect{h}_k^T  \vect{w}_k) + \sum \limits_{j=1, j\neq k}^K  c_1^2 \beta_k  E[ |\vect{h}_k^T  \vect{w}_j|^2]\\
	& + c_2^2 \beta_k E[ \vect{h}_k^T \vect{S} \vect{h}_k^{\ast}]+ \beta_k \sigma_q^2 p_d   +1.
	\end{align}
where $\vect{w}_j$ is the $j$-th vector of the beamforming matrix $\vect{W}$ and $\vect{S}$ is the artificial-noise shaping matrix as defined previously.

In the proofs, we write the channel vector $\vect{h}_k $ as a sum of channel estimate and uncorrelated estimation error, i.e., 
\begin{equation}
\vect{h}_k = \hat{\vect{h}}_k+ \vect{e}_k
\end{equation}
where $\hat{\vect{h}}_k \in \mathcal{C}^N$ is the channel estimate with i.i.d. $\mathcal{CN}(0,\sigma_{\hat{h}_{k}}^2)$ components, and $\vect{e}_k \in \mathcal{C}^N$ is uncorrelated estimation error with i.i.d. $\mathcal{CN}(0,1-\sigma_{\hat{h}_{k}}^2)$ components, i.e.,  $E\left[\vect{e}_k\vect{e}_k^H\right]= (1-\sigma_{\hat{h}_{k}}^2) \vect{I}_N$.

In the following, we will evaluate \eqref{eq10.2_1} for MRT-BF and ZF-BF. For each beamforming scheme, we evaluate  the deterministic constant $a$ in \eqref{constant_a} and variance of effective noise $\sigma_{n_{\text{eff}}}^2$ in \eqref{var_noise}. 

%==========================================================
\subsection{MRT-BF}
\label{proof_thm1_MRT}
%==========================================================

From \eqref{eq:mrt_zf_bf}, the MRT-BF matrix is given by $\vect{W}_{\text{mrt}}:=\vect{W}=\widehat{\vect{H}}^{\ast}$. Let  $\vect{w}_{\text{mrt}, j} = \widehat{\vect{h}}_j^{\ast} $ be the $j$-th column of $\vect{W}_{\text{mrt}}$, i.e., the beamforming vector of user $j$. In the following, we will evaluate \eqref{constant_a} and \eqref{var_noise} and then substitute the results in  \eqref{eq10.2_1}.  

From \eqref{constant_a}, we have

	\begin{align} \label{eq_appA1}
	\nonumber
	a_{\textnormal{mrt}}&=c_1 \sqrt{\beta_{k}} \underbrace{E[ \vect{h}_k^T   \vect{w}_{\text{mrt}, k}]}_{I_0} =c_1 \sqrt{\beta_{k}} E[ \vect{h}_k^T   \hat{\vect{h}}_k^{\ast}] \\ \nonumber
	 &=  c_1 \sqrt{\beta_{k}} E[ (\hat{\vect{h}}_k+ \vect{e}_k)^T   \hat{\vect{h}}_k^{\ast}] \\ \nonumber
&=	 c_1 \sqrt{\beta_{k}} E[ \norm{\hat{\vect{h}}_k}^2] +c_1 \sqrt{\beta_{k}}  \underbrace{E[ \vect{e}_k^T   \hat{\vect{h}}_k^{\ast}]}_{=0}\\
&=c_1 \sqrt{\beta_{k}}  \underbrace{N \sigma_{\hat{h}_{k}}^2}_{I_0}
	\end{align}
	and hence 
	\begin{equation} \label{eq_appA1.1}
	|a_{\textnormal{mrt}}|^2= \frac{2 \theta \beta_k p_d}{\pi \operatorname{tr}(\vect{\Sigma})} \sigma_{\hat{h}_{k}}^4 N.
	\end{equation}
	%-----------------------
	
	Using \eqref{var_noise} we write 
	\begin{align} 
	\label{eq_appA2} \nonumber
	&\sigma_{n_{\text{eff}}, \textnormal{mrt} }^2 = \\\nonumber
	& c_1^2 \beta_k\operatorname{Var} (\vect{h}_k^T  \vect{w}_{\text{mrt}, k}  )+ \sum \limits_{j=1, j\neq k}^K  c_1^2 \beta_k  E[ |\vect{h}_k^T  \vect{w}_{\text{mrt}, j}  |^2] \\\nonumber
	& + c_2^2 \beta_k E[ \vect{h}_k^T \vect{S} \vect{h}_k^{\ast}]+ \beta_k \sigma_q^2 p_d   +1\\\nonumber
	&=c_1^2 \beta_k\underbrace{\operatorname{Var} (\vect{h}_k^T  \hat{\vect{h}}_k^{\ast} )}_{I_1}+ \sum \limits_{j=1, j\neq k}^K  c_1^2 \beta_k  \underbrace{ E[ |\vect{h}_k^T  \hat{\vect{h}}_j^{\ast} |^2]}_{I_2} \\
	& + c_2^2 \beta_k\underbrace{E[ \vect{h}_k^T \vect{S} \vect{h}_k^{\ast}]}_{I_3} + \beta_k \sigma_q^2 p_d   +1.
	\end{align}
	The terms $I_1, I_2$ can be obtained as follows:
	\begin{align} 
	\label{eq_appA3_1} 
	\nonumber
	I_1 &\stackrel{}{=} E[ |\vect{h}_k^T   \hat{\vect{h}}_k^{\ast}|^2]-|I_0|^2\\ \nonumber
	&\stackrel{a}{=} E[ |(\hat{\vect{h}}_k+ \vect{e}_k)^T   \hat{\vect{h}}_k^{\ast}|^2] - N^2\sigma_{\hat{h}_{k}}^4 \\ 
	&\stackrel{b}{=}  E[\norm{\hat{\vect{h}}_k}^4+E[|\vect{e}_k^T \hat{\vect{h}}_k^{\ast}|^2]]-N^2\sigma_{\hat{h}_{k}}^4=N \sigma_{\hat{h}_{k}}^2 \\ \label{eq_appA3_2}  \nonumber
	I_2 &\stackrel{c}{=}  E[ \vect{h}_k^T  \hat{\vect{h}}_j^{\ast} \hat{\vect{h}}_j^{T} \vect{h}_k^{\ast}]=E[ \vect{h}_k^T  E[\hat{\vect{h}}_j^{\ast} \hat{\vect{h}}_j^{T}] \vect{h}_k^{\ast}]\\
	&=   \sigma_{\hat{h}_{j}}^2 E[ \norm{\vect{h}_k}^2] = N \sigma_{\hat{h}_{j}}^2.
	\end{align}
	where in (a) we use $I_0=N\sigma_{\hat{h}_{k}}^2$ evaluated in \eqref{eq_appA1}, in (b) we use the fact that $\vect{h}_k$ and $\hat{\vect{e}}_k$ are uncorrelated and $E[\norm{\hat{\vect{h}}_k}^4=N(N+1) \sigma_{\hat{h}_{k}}^2$ \cite{T2018}, and in (c) we make use of the statistical independence of $\vect{h}_k$ and $\hat{\vect{h}}_j$.
	
	%---------------------------------
	Regarding $I_3$:
	
	\textbf{Case 1:} From \eqref{eq_shaping} ,  $\vect{S} = \vect{I}_N$ when R-AN scheme is used, leading to 
	\begin{equation}
	I_3 = E[ \vect{h}_k^T \vect{S} \vect{h}_k^{\ast}]=E[ \norm{\vect{h}_k}^2]=N.
	\end{equation}

\textbf{Case 2:}  Again, from \eqref{eq_shaping} we have $\vect{S} = 	\vect{I}_N - \vect{P}_{\text{proj}}$ when NS-AN scheme is used, where $ \vect{P}_{\text{proj}} = \widehat{\vect{H}}^{\ast} (\widehat{\vect{H}}^T \widehat{\vect{H}}^{\ast})^{-1}\widehat{\vect{H}}^T$ is the projection matrix. Hence,
%--------------------------------
\begin{align} \nonumber
I_3 &= E[ \vect{h}_k^T (\vect{I}_N - \vect{P}_{\text{proj}} )\vect{h}_k^{\ast}]\\ \nonumber
&=E[ \norm{\vect{h}_k}^2] - E[ \vect{h}_k^T \vect{P}_{\text{proj}}\vect{h}_k^{\ast}]\\ \nonumber
&=N - E[ (\hat{\vect{h}}_k+ \vect{e}_k)^T \vect{P}_{\text{proj}} (\hat{\vect{h}}_k+ \vect{e}_k)^{\ast}]\\\nonumber
&=N - E[ (\hat{\vect{h}}_k^T \vect{P}_{\text{proj}} \hat{\vect{h}}_k^{\ast}]  - E[ \vect{e}_k^T \vect{P}_{\text{proj}} \vect{e}_k^{\ast}] \\\nonumber
&\stackrel{a}{=}N-E[ \norm{\hat{\vect{h}}_k}^2]- E [\operatorname{tr} (  \vect{P}_{\text{proj}} \vect{e}_k^{\ast} \vect{e}_k^T )] \\\nonumber
&\stackrel{b}{=}N-N\sigma_{\hat{h}_{k}}^2-E [\operatorname{tr} (  \vect{P}_{\text{proj}} E [\vect{e}_k^{\ast} \vect{e}_k^T] )]  \\  \nonumber
&\stackrel{c}{=} N-N\sigma_{\hat{h}_{k}}^2- K (1-\sigma_{\hat{h}_{k}}^2)\\
&=(N-K)(1-\sigma_{\hat{h}_{k}}^2)
\end{align}
where (a) follows due to  $\vect{P}_{\text{proj}} \hat{\vect{h}}_k^{\ast}=\hat{\vect{h}}_k^{\ast}$ (by definition), i.e., channel estimate is projected onto itself  (b) due to the statistical independence between channel estimate and estimation error, and in (c) we make use of $E[\operatorname{tr} ( \vect{P}_{\text{proj}})]=E[ \operatorname{tr} ( \vect{I}_K) ]= K$.  
%-------------------------------------------
	
Summarizing the above results for $I_3$, we write 
	\begin{equation}\label{eq_appA4}
	I_3=
	\begin{cases}
	N &\text{if R-AN }  \\
	(N-K) (1-\sigma_{\hat{h}_{k}}^2) &\text{if NS-AN}.
	\end{cases}
	\end{equation} 
Substituting~ \eqref{eq_appA3_1}, \eqref{eq_appA3_2} and \eqref{eq_appA4} with  definitions of $c_1$ and $c_2$  in \eqref{eq_appA2}  yields
	\begin{align} 
	\label{eq_appA5} 
	\nonumber
	\sigma_{n_{\text{eff}}, \textnormal{mrt}}^{2}&= \underbrace{\frac{2\theta \beta_k p_d}{\pi \operatorname{tr}(\vect{\Sigma})} \sigma_{\hat{h}_{k}}^2}_{\text{beamforing gain penalty}} + \underbrace{\frac{2\theta \beta_k p_d}{\pi \operatorname{tr}(\vect{\Sigma})}\sum_{j=1,j\neq k}^K \sigma_{\hat{h}_{j}}^2}_{\text{inter-user interference}}\\\nonumber
	& + \underbrace{ P^{\textnormal{AN}}_k  }_{\text{artificial noise}} +\underbrace{\beta_k \sigma_q^2 p_d +1}_{\text{quantization noise plus AWGN}} \\ 
	&=\frac{2  \theta \beta_k p_d}{\pi} + P^{\text{AN}}_k +\beta_k \sigma_q^2 p_d +1. 
	\end{align}
	where 
	\begin{equation} \label{eq_appA6}
P^{\textnormal{AN}}_k =
	\begin{cases}
	{2 \bar{\theta} \beta_k p_d }/{\pi}&\text{if R-AN }  \\
	{2 \bar{\theta} \beta_k p_d }(1-\sigma_{\hat{h}_{k}}^2)/\pi &\text{if NS-AN}.
	\end{cases}
	\end{equation} 
	In \eqref{eq_appA5}, the \textit{beamforming gain penalty} is due to the CSI uncertainty at the user. Substituting \eqref{eq_appA1.1}, \eqref{eq_appA5} with \eqref{eq_appA6} in \eqref{eq10.2_1}, the first part of Theorem \ref{thm1} follows.

%==========================================================
	\subsection{ZF-BF}
\label{proof_thm1_ZF}
%==========================================================
	
	From \eqref{eq:mrt_zf_bf}, the ZF-BF matrix is given by $\vect{W}_{\text{zf}}:=\vect{W}=\widehat{\vect{H}}^{\ast} (\widehat{\vect{H}}^T \widehat{\vect{H}}^{\ast})^{-1}$ satisfying $\widehat{\vect{H}}^{T} \vect{W}_{\text{zf}}= \vect{I}_K$. Let  $\vect{w}_{\text{zf}, j}$ be the $j$-th column of $\vect{W}_{\text{zf}}$, i.e., the beamforming vector of user $j$.  In our analysis we need the following Lemma.
	\begin{lemma} \label{zf_norm}
	The expected value of norm squared of $j$-th column of $\vect{W}_{\textnormal{zf}}$  is given by
	\begin{equation}
	E[\norm{\vect{w}_{\textnormal{zf},j}}^2] = \frac{\sigma_{\hat{h}_j}^{-2}}{N-K}
	\end{equation}
	\end{lemma}
	\begin{proof}
	Using the channel estimation decomposition in \eqref{eq:CE_decomposition}, i.e., $\widehat{\vect{H}} = \widetilde{\vect{H}} \vect{\Sigma}^{1/2}$, the ZF-BF matrix $\vect{W}_{\text{zf}} $ can be expressed in terms of $\widetilde{\vect{H}}$ and $\vect{\Sigma}$ as follows:
	\begin{align}
	\nonumber
	\vect{W}_{\text{zf}} &= \widehat{\vect{H}}^{\ast} (\widehat{\vect{H}}^T \widehat{\vect{H}}^{\ast})^{-1}\\
	&= \widetilde{\vect{H}}^{\ast} (\widetilde{\vect{H}}^T \widetilde{\vect{H}}^{\ast})^{-1}\      \vect{\Sigma}^{-1/2}
	\end{align}
	Thus we can write 
		\begin{align}
		\nonumber
	E[\norm{\vect{w}_{\textnormal{zf},j}}^2]&= [\vect{W}_{\text{zf}} ^H \vect{W}_{\text{zf}} ]_{j} \\ \nonumber
	& = \left[  \vect{\Sigma}^{-1/2}  \underbrace{E[(\widetilde{\vect{H}}^T \widetilde{\vect{H}}^{\ast})^{-1} ] }_{=(N-K)^{-1} \vect{I}_K}   \vect{\Sigma}^{-1/2}  \right]_j \\
&= \left [\frac{\vect{\Sigma}^{-1}}{N-K}  \right]_j=\frac{\sigma_{\hat{h}_j}^{-2}}{N-K}
	\end{align}
where $[\vect{A}]_j$ is the $ j$-th diagonal entry of $\vect{A}$. The inner expectation in the second line follows from the properties of the $K \times K$ central Wishart matrix $\widetilde{\vect{H}}^T \widetilde{\vect{H}}^{\ast}$ \cite{RMT_2004}.
	\end{proof}

	In the following, we will evaluate \eqref{constant_a} and \eqref{var_noise} and then substitute the results in  \eqref{eq10.2_1} to obtain the achievable rate under ZF-BF.  
	
From \eqref{constant_a} we have
	\begin{align} \label{eq:eval_J0}
	\nonumber
	a_{\textnormal{zf}} &= c_1 \sqrt{\beta_{k}} \underbrace{E[ \vect{h}_k^T   \vect{w}_{\text{zf}, k}]}_{J_0}= c_1 \sqrt{\beta_{k}} E[ (\hat{\vect{h}}_k+ \vect{e}_k)^T   \vect{w}_{\text{zf}, k}] \\ \nonumber
	&= c_1 \sqrt{\beta_{k}} (E[ \underbrace{\hat{\vect{h}}_k^T \vect{w}_{\text{zf}, k}}_{=1 (\text{by definition})}] + \underbrace{E[\vect{e}_k^T \vect{w}_{\text{zf}, k} ] }_{ =0  (\text{uncorrelated})}) \\
	&=c_1\sqrt{\beta_{k}}
	\end{align}
	and hence
	\begin{equation} \label{eq_appAA2}
	|a_{\textnormal{zf}}|^2 = \frac{2 \theta \beta_k p_d}{\pi \operatorname{tr}(\vect{\Sigma}^{-1})}(N-K)
	\end{equation}
	
		Using \eqref{var_noise} we write 
	\begin{align}
	\label{eq_appAA0}  \nonumber
	&\sigma_{n_{\textnormal{eff}}, \textnormal{zf}}^{2} = c_1^2 \beta_k\underbrace{\operatorname{Var} (\vect{h}_k^T \vect{w}_{\text{zf}, k} )}_{J_1}+ \sum \limits_{j=1, j\neq k}^K  c_1^2 \beta_k  \underbrace{ E[ |\vect{h}_k^T  \vect{w}_{\text{zf}, j} |^2]}_{J_2} \\
	& + c_2^2 \beta_k\underbrace{E[ \vect{h}_k^T \vect{S} \vect{h}_k^{\ast}]}_{J_3=I_3} + \beta_k \sigma_q^2 p_d   +1.
	\end{align}
  Note that we need to evaluate $J_1$ and $J_2$, while $J_3=I_3$ is given in \eqref{eq_appA4}. 
  
  To evaluate the terms $J_1$ and $J_2$, we proceed as follows.
	%=================================
	\begin{align} \label{eq_appAA3} \nonumber
	J_1&\stackrel{a}{=}  E[ |\vect{h}_k^T   \vect{w}_{\text{zf}, k}|^2] -|\underbrace{E[ \vect{h}_k^T   \vect{w}_{\text{zf},k}]}_{J_0}|^2\\ \nonumber
	& \stackrel{b}{=}E[ |\vect{h}_k^T  \vect{w}_{\text{zf},k}|^2]-1=E[ |\underbrace{{\hat{\vect{h}}_k}^T \vect{w}_{\text{zf},k}}_{=1 } + \vect{e}_k^T  \vect{w}_{\text{zf},k}|^2]-1\\ \nonumber
	&\stackrel{c}{=}E[|{\vect{e}}_k^T\vect{w}_{\text{zf},k}|^2]  = (1-\sigma_{\hat{h}_{k}}^2) E[\norm{\vect{w}_{\text{zf},k}}^2]\\
	&\stackrel{d}{=}\frac{(1-\sigma_{\hat{h}_{k}}^2)\sigma_{\hat{h}_{k}}^{-2}}{N-K}
	\end{align}
where in (b) we use $J_0=1$ evaluated in \eqref{eq:eval_J0}, (b) follows  because $\vect{e}_k$ and $\vect{w}_{\text{zf},k}$ are independent and ${\hat{\vect{h}}_k}^T \vect{w}_{\text{zf},k}=1$ (by definition), and in (d) we use Lemma \ref{zf_norm}.

Next,
	%==================================
	\begin{align} \label{eq_appAA4} \nonumber
	J_2 &\stackrel{a}{=} E[ |\vect{h}_k^T  \vect{w}_{\text{zf}, j}|^2] =E[|(	\hat{\vect{h}}_k+ \vect{e}_k)^T \vect{w}_{\text{zf}, j}|^2] \\ \nonumber
	&\stackrel{b}{=} E[| \underbrace{\hat{\vect{h}}_k^T \vect{w}_{\text{zf}, j}}_{=0} + \vect{e}_k^T \vect{w}_{\text{zf}, j} |^2]=E[| \vect{e}_k^T \vect{w}_{\text{zf}, j} |^2] \\\nonumber
	&\stackrel{c}{=} (1-\sigma_{\hat{h}_{k}}^2) E[\norm{\vect{w}_{\text{zf}, j}}^2]\\
	&\stackrel{d}{=} \frac{(1-\sigma_{\hat{h}_{k}}^2)\sigma_{\hat{h}_{j}}^{-2}}{N-K}
	\end{align}
where in (b) $\hat{\vect{h}}_k^T \vect{w}_{\text{zf}, j}=0$ follows by definition of zero-forcing solution, (c) follows  because $\vect{e}_k$ and $\vect{w}_{\text{zf},j}$ are independent  and (d) uses Lemma \ref{zf_norm}.

	Substituting \eqref{eq_appAA3}, \eqref{eq_appAA4} and \eqref{eq_appA4} with the definitions of $c_1$ and $c_2$ in \eqref{eq_appAA0}, the variance of the effective noise can be expressed by
	\begin{align}\label{eq_appAA5}\nonumber
	\sigma_{n_{\text{eff}}, \textnormal{zf} }^{2}&= \underbrace{\frac{2\theta \beta_k p_d (1-\sigma_{\hat{h}_{k}}^2) }{\pi  \sigma_{\hat{h}_{k}}^{2} \operatorname{tr}(\vect{\Sigma}^{-1})}}_{\text{ beamforing gain penalty}} +
	\underbrace{\frac{2\theta \beta_k p_d}{\pi  \operatorname{tr}(\vect{\Sigma}^{-1})}\sum_{j=1,j\neq k}^K\frac{1-\sigma_{\hat{h}_{k}}^2}{\sigma_{\hat{h}_{j}}^{2}} }_{\text{inter-user interference}}\\\nonumber
	& +  \underbrace{ P^{\text{AN}}_k  }_{\text{artificial noise}}+\underbrace{\beta_k \sigma_q^2 p_d +1}_{\text{quantization noise plus AWGN}} \\ 
	&=\frac{2 \theta \beta_k p_d }{\pi} (1-\sigma_{\hat{h}_{k}}^2)  +P^{\text{AN}}_k + \beta_k \sigma_q^2 p_d +1. 
	\end{align}
	
	Finally, substituting \eqref{eq_appAA2}, \eqref{eq_appAA5} combined with \eqref{eq_appA6} in \eqref{eq10.2_1}, the second part of Theorem \ref{thm1} follows. This completes the proof.

    %==========================================================
	\section{Proof of Theorem 2}
	 \label{proof_thm2}
	 %==========================================================
	 
	Here we derive the upper bounds \eqref{thm2_eq1}, \eqref{thm2_eq2} on the information rate $\overline{R}_{e}$ leaked to the eavesdropper under 
	MRT-BF and ZF-BF schemes. 
	
	From Lemma \ref{R_e_general},  by the concavity of $\log (\cdot)$, applying Jensen's inequality to \eqref{eq:R_e_general} yields
	\begin{equation} \label{eq_thm2_JIE}
	\overline{R}_{e} \le  \log \left ( 1 + {c_1^2 \beta_{e} E\left [\vect{w}_k^H \vect{g}^{\ast}  \sigma_{e}^{-2} \vect{g}^T \vect{w}_k \right ] }  \right).
	\end{equation}
	where $\sigma_{e}^{2}$ is the variance of effective noise given in Lemma \ref{R_e_general}, which is rewritten again here:
	\begin{equation}\label{eq_cov_eavesdropper_rep}
\sigma_{e}^{2} =  c_2^2 \beta_{e} \vect{g}^T \vect{S} \vect{g}^{\ast}+ c_3^2 \beta_{e}\sigma_q^2  \vect{g}^T  \vect{g}^{\ast}+ 1.
	\end{equation}
Since $\sigma_{e}^{2} $ is dependent of artificial noise scheme (R-AN or NS-AN), in the following we evaluate \eqref{eq_cov_eavesdropper_rep} for R-AN and NS-AN schemes, respectively.\\

\textbf{Case 1}: When R-AN approach is used, from \eqref{eq_shaping}  we have $\vect{S} = \vect{I}_N$. Hence,
	%===========================================
	\begin{align} \label{thm2_proof_1} \nonumber
\sigma_{e}^{2, \textnormal{R-AN}}& =  (c_2^2 \beta_{e} + c_3^2 \beta_{e}\sigma_q^2) \norm{\vect{g}}^2 + 1\\ \nonumber
	&\xrightarrow{\text{a.s.}}  (c_2^2 \beta_{e} + c_3^2 \beta_{e}\sigma_q^2) N +1 \\
	&= {2 \bar{\theta} \beta_{e} p_d}/\pi + \beta_{e} p_d \sigma_q^2 + 1
	\end{align} 
	as $N$ grows large which follows from the strong law of large numbers.\\
	
	\textbf{Case 2}:  When NS-AN approach is used, from \eqref{eq_shaping}, we have $\vect{S} = 	\vect{I}_N -\vect{ P}_{\text{proj}}$. Using Lemma \ref{lemma1}, we can write
	%============================================
	\begin{align} \nonumber
	\vect{g}^T \vect{S} \vect{g}^{\ast} &= ( \sqrt{\kappa_R} \hat{\vect{h}}_k+ \vect{\epsilon})^T \vect{S} ( \sqrt{\kappa_R}  \hat{\vect{h}}_k+ \vect{\epsilon})^{\ast} \\ \label{thm2_proof_3}
	&= \vect{\epsilon}^T \vect{S} \vect{\epsilon}^{\ast}=\vect{\epsilon}^T \widetilde{\vect{U}}\widetilde{\vect{U}}^H\vect{\epsilon}^{\ast}
	\end{align}
	where $\widetilde{\vect{U}} \in \mathcal{C}^{N \times (N-K)}$ comprise $(N-K)$ eigenvectors (each has norm 1) corresponding to the $N-K$ repeated unity eigenvalues of $S$. Since $N\gg K$ (i.e., Massive MIMO setting), $\widetilde{\vect{U}}\widetilde{\vect{U}}^H$ can be very well approximated by a scaled identity matrix, where the magnitude of off-diagonal entries of $\widetilde{\vect{U}} \widetilde{\vect{U}} ^H$ are in fact much smaller than the diagonal entries. Thus,
	%===========================================
	\begin{equation} \label{thm2_proof_5}
	\widetilde{\vect{U}}\widetilde{\vect{U}}^H \approx \frac{\operatorname{tr}( \widetilde{\vect{U}}\widetilde{\vect{U}}^H)}{N} = \left(1-\frac{K}{N} \right) \vect{I}_N
	\end{equation}
	
	Substituting \eqref{thm2_proof_5} in \eqref{thm2_proof_3} yields
	%===========================================
	\begin{align} 
	\vect{g}^T \vect{S}\vect{g}^{\ast} \approx  (1-K/N) \norm{\vect{\epsilon}}^2 \xrightarrow{\text{a.s.}}  ({N-K}) (1-\kappa_R \sigma_{\hat{h}_{k}}^2 ) 
	\end{align}
	Therefore,
	%============================================
	\begin{equation}
\sigma_e^{2, \textnormal{NS-AN}}\approx {2 \bar{\theta} \beta_{e} p_d} (1-\kappa_R \sigma_{\hat{h}_{k}}^2 )/\pi  + \beta_{e} p_d \sigma_q^2 + 1.
	\end{equation}
	We summarize, 
	%============================================
	\begin{equation} \label{eq_eavesdropper_cov_approx}
	\sigma_{e}^{2} \approx (P_{e}^{\text{AN}} + \beta_{e} p_d \sigma_q^2 + 1) 
	\end{equation}
	where 
	%============================================
	\begin{equation} \label{eavesdropper_noise_def}
	P_{e}^{\text{AN}} =
	\begin{cases}
	2\bar{\theta}  \beta_{e} p_d/\pi  &\text{if R-AN }  \\
	{2 \bar{\theta} \beta_{e} p_d }(1-\kappa_R \sigma_{\hat{h}_{k}}^2)/\pi &\text{if NS-AN}.
	\end{cases}
	\end{equation} 
	
	Substituting \eqref{eq_eavesdropper_cov_approx} in \eqref{eq_thm2_JIE} yields
	
	\begin{equation} \label{eavesdropper_rate_asymp}
	\overline{R}_{e}  \lesssim  \log \left ( 1 + \frac{c_1^2 \beta_{e} E[\vect{w}_k^H \vect{g}^{\ast} \vect{g}^T \vect{w}_k ] }{P_{e}^{\text{AN}} + \beta_{e} p_d \sigma_q^2 + 1}  \right) .
	\end{equation}
	
	The expectation $\mu := E[\vect{w}_k^H \vect{g}^{\ast} \vect{g}^T \vect{w}  ]$ for both the MRT-BF and ZF-BF cases is evaluated as follows. 

For MRT-BF, setting $\vect{w}_{\text{mrt},k}: =\vect{w}_k= \hat{\vect{h}}_k^{\ast}$ (i.e., $k$-th column of MRT-BF matrix $\vect{W}_{\text{mrt}} = \widehat{\vect{H}}^{\ast}$ given in \eqref{eq:mrt_zf_bf}). Using Lemma \ref{lemma1}, we write 
	\begin{align} \label{mu_mrt}\nonumber
	&\mu_{\text{mrt}} :=E \left [ \vect{w}_{\text{mrt},k}^{H}  \vect{g}^{\ast} \vect{g}^{T} \vect{w}_{\text{mrt},k} \right] =E\left [ \hat{\vect{h}}_k^{T}  \vect{g}^{\ast} \vect{g}^{T} \hat{\vect{h}}_k^{\ast} \right] \\ \nonumber 
	&=E\left [ \hat{\vect{h}}_k^{T}  (\sqrt{\kappa_R} \hat{ \vect{h}}_k^{\ast} + \vect{\epsilon}^{\ast})(\sqrt{\kappa_R} \hat{ \vect{h}}_k^{T} + \vect{\epsilon}^{T})\hat{\vect{h}}_k^{\ast} \right] \\ \nonumber 
&= \kappa_R E  [ \norm{\hat{\vect{h}}_k}^4]  + 2\sqrt{\kappa_R} \Re \{ \underbrace{E[\hat{\vect{h}}_k^T \hat{\vect{h}}_k^{\ast} \vect{\epsilon}^{T}  \hat{\vect{h}}_k^{\ast}]}_{=0}\}+  E[\hat{\vect{h}}_k^T {\vect{\epsilon}}^{\ast} \vect{\epsilon}^{T} \hat{\vect{h}}_k^{\ast}]\\ \nonumber
	&=  \kappa_R \sigma_{\hat{h}_k}^4 N(N+1) +  (1-\kappa_R \sigma_{\hat{h}_{k}}^2) \sigma_{\hat{h}_{k}}^2 N \\
	&=  \sigma_{\hat{h}_{k}}^2 (\kappa_R \sigma_{\hat{h}_{k}}^2 N +1)N.
	\end{align}
where in the fourth line we make use of the fact that $\vect{\epsilon}$ is independent of $\hat{\vect{h}}$ each with zero-mean and $E[\norm{\hat{\vect{h}}_k}^4=N(N+1) \sigma_{\hat{h}_{k}}^2$ \cite{T2018}.
	
For ZF-BF, setting $\vect{w}_{\text{zf}, k} := \vect{w}_k$ as the $k$-th column of ZF-BF  matrix \eqref{eq:mrt_zf_bf} given by $\vect{W}_{\text{zf}} = \widehat{\vect{H}}^{\ast} (\widehat{\vect{H}}^T \widehat{\vect{H}}^{\ast})^{-1} $. Then by using Lemma \ref{lemma1} we can write 	%============================================
	\begin{align} \label{mu_zf}
	\nonumber
	&\mu_{\text{zf}}: =  E\left [\vect{w}_{\text{zf}, k} ^H \vect{g}^{\ast} \vect{g}^{T} \vect{w}_{\text{zf}, k} \right]\\ \nonumber
	&=E\left[  \vect{w}_{\text{zf}, k} ^H (\sqrt{\kappa_R} \hat{ \vect{h}}_k^{\ast} + \vect{\epsilon}^{\ast}) (\sqrt{\kappa_R} \hat{ \vect{h}}_k^{T} + \vect{\epsilon}^{T}) \vect{w}_{\text{zf}, k}  \right] \\ \nonumber
	&= \kappa_R E  [ \norm{ \underbrace{\vect{w}_{\text{zf}, k} ^H \vect{h}_k^{\ast}}_{=1}  }^2 ]+ 2\sqrt{\kappa_R} \Re \{ \underbrace{E[ \vect{w}_{\text{zf}, k}^H \hat{\vect{h}}_k^{\ast} \vect{\epsilon}^{T}   \vect{w}_{\text{zf}, k} ]}_{=0}\} \\ \nonumber
	&+E[\vect{w}_{\text{zf}, k} ^H \vect{\epsilon}^{\ast} \vect{\epsilon}^{T}  \vect{w}_{\text{zf}, k} ] =\kappa_R  + E[\vect{w}_{\text{zf}, k} ^H  E [\vect{\epsilon}^{\ast} \vect{\epsilon}^{T} ] \vect{w}_{\text{zf}, k} ] \\ \nonumber
	&= \kappa_R  + (1-\kappa_R \sigma_{\hat{h}_{k}}^2)  E[ \norm{\vect{w}_{\text{zf}, k} }^2]  \\
	&=  \kappa_R  +   \frac{\sigma_{\hat{h}_{k}}^{-2} (1-\kappa_R \sigma_{\hat{h}_{k}}^2) }{N-K} 
	\end{align}
	%Therefore,
	%%==============================
where in the derivation steps of \eqref{mu_zf} we have used the zero-forcing property: $\vect{w}_{\text{zf}, k}^H \vect{h}_k^{\ast}  = \vect{h}_k^{T} \vect{w}_{\text{zf}, k}=1$, statistical independence of $\vect{\epsilon}$ and $\vect{w}_k$, and in the last line we use Lemma \ref{zf_norm}.

Finally, substituting \eqref{mu_mrt} and \eqref{mu_zf} combined with \eqref{eavesdropper_noise_def} and the definition of $c_1$ in \eqref{eavesdropper_rate_asymp}, \eqref{thm2_eq1} and \eqref{thm2_eq2} follow, respectively. In \eqref{thm2_eq1} and \eqref{thm2_eq2} the notation $ \lesssim$ is replaced by $ \cong $ where the notion of upper bound is understood from the bars over the symbols. This completes the proof.
\end{appendices} 

%\section*{Acknowledgment}
%The author would like to thank the reviewers for their suggestive comments that have contributed to the improvement of the work in its current form.
%-------------------------Bibilography-----------------------------------------
\bibliographystyle{IEEEtran}
\bibliography{IEEEabrv,Ref}

%--------------------------Short Bio-------------------------------------------

%\begin{IEEEbiography}[{\includegraphics[width=1in,height=1.25in,clip,keepaspectratio]{titi.jpg}}]{M. A. Teeti} received the Ph.D. degree in communication and information systems from Huazhong University of Science and Technology, Wuhan, China, in 2015. He is currently an Associate Professor with the Department of Communication and Information Engineering, East China University of Technology, Nanchang, China. His research interests are in the fields of wireless communication and information theory, with a focus on Massive MIMO, mm-wave MIMO, and physical layer security.
%
%In 2016, he joined the Department of Electrical and Electronic Engineering, Southern University of Science and Technology, Shenzhen, China, as a Visiting Scholar and later as a Postdoctoral Fellow. From 2014 to 2015, he was a Visiting Research Student with the InfoLab21, Lancaster University, UK. In 2010, he received the Ph.D. Distinguished International Student Scholarship from China Scholarship Council, Beijing, China. From 2006 to 2008, he served as a Technical Instructor with the University College for Educational Sciences, Ramallah, Palestine. From 2003 to 2006, he was a Teaching and Research Assistant with the Faculty of Engineering, Al-Quds University, Abu-Dis, Jerusalem.
%\end{IEEEbiography}
%\EOD
\end{document}